\documentclass[sn-mathphys-num]{sn-jnl}

\usepackage{graphicx}%
\usepackage{multirow}%
\usepackage{amsthm}%
\usepackage{mathrsfs}%
\usepackage[title]{appendix}%
\usepackage{xcolor}%
\usepackage{textcomp}%
\usepackage{manyfoot}%
\usepackage{booktabs}%
\usepackage{algpseudocode}%
\usepackage{listings}%

\usepackage{framed}
\usepackage{url}
\usepackage[ruled,vlined,linesnumbered]{algorithm2e}
\usepackage{comment}
\usepackage{tikz}
\usetikzlibrary{positioning,quotes}
\usepackage[utf8]{inputenc}
\usepackage{array}
\usepackage{color, colortbl}
\usepackage[T1]{fontenc}
\usepackage{todonotes}
\usepackage{float}
\usepackage{comment}

\usepackage{soul}
\usepackage{multirow}
\usepackage[export]{adjustbox}
\usepackage{amsfonts} 

\usepackage{lmodern}

\usepackage{lineno}

\usepackage[leqno]{amsmath}
\usepackage{threeparttable}
\usepackage{amsmath}
\usepackage{multirow}
\usetikzlibrary{arrows,shapes}
\usepackage{algorithmicx}
\newtheorem{lemma}{Lemma}
\newcommand{\df}{\textsc{Diffusion}}
\newcommand{\LBIF}{\textsc{Level-Based-Influence}}
\newcommand{\PBGF}{\textsc{Profit-Based-Greedy}}
\newcommand{\MPBF}{\textsc{Max-Profit-Based}}
\newcommand{\spr}{\textsc{Spreader}}

\newcommand{\LB}{\text{{\scriptsize{LBGH}}}}
\newcommand{\MD}{\text{{\scriptsize{MDFH}}}}
\newcommand{\PB}{\text{{\scriptsize{PBGH}}}}
\newcommand{\MP}{\text{{\scriptsize{MPBGH}}}}
\newcommand{\CB}{\text{{\scriptsize{CBH}}}}

\newcommand{\LBE}{\text{Level Based Greedy Heuristic}}
\newcommand{\MDE}{\text{Maximum Degree First Heuristic}}
\newcommand{\PBE}{\text{Profit Based Greedy Heuristic}}
\newcommand{\MPE}{\text{Maximum Profit Based Greedy Heuristic}}

\newcommand\clearrow{\global\let\rowmac\relax}

\newcommand{\IM}{\textsc{Interest  Maximization}}

\newcommand{\GR}{\textbf{\scriptsize{$G(V, E, \eta,t_A)$}}}

\newcounter{cases}
\newcounter{subcases}[cases]

\usepackage{geometry}
\usepackage{multirow,tabularx}

\newcolumntype{Y}{>{\centering\arraybackslash}X}
\renewcommand{\arraystretch}{2}

\SetKwComment{Comment}{$\triangleright$\ }{}

\tikzstyle{vertex}=[circle,fill=black!25,minimum size=12pt,inner sep=0pt]
\tikzstyle{selected vertex} = [vertex, fill=red!60]
\tikzstyle{edge} = [draw,thick,-]
 \tikzstyle{weight} = [font=\small]
\tikzstyle{selected edge} = [draw,line width=3pt,-,red!50]
\tikzstyle{ignored edge} = [draw,line width=3pt,-,black!20]

\let\cline\cmidrule

\begin{document}
	
	\title[Interest Maximization in Social Networks]{Interest Maximization in Social Networks}
	
	
	\author*[$^1$]{\fnm{Rahul Kumar} \sur{ Gautam$^{0009-0009-7693-3863}$ }}\email{19mcpc06@uohyd.ac.in}
	
	\author[$^2$]{\fnm{Anjeneya Swami} \sur{Kare$^{0000-0003-3644-4802}$}}\email{askcs@uohyd.ac.in}
	
	\author[$^3$]{\fnm{S. Durga} \sur{Bhavani$^{0000-0003-4413-0328}$}}\email{sdbcs@uohyd.ac.in}
	
	\affil[]{\orgdiv{School of Computer and Information Science}, \orgname{University Of Hyderabad}, \orgaddress{\street{Gachibowli}, \city{Hyderabad}, \postcode{500046}, \state{Telangana}, \country{India}}}
	
	
	
	
	\abstract{Nowadays, organizations use viral marketing strategies to promote their products through social networks. A graph represents the social networks.
		It is expensive to directly send the product promotional information to all the users in the network. In this context, Kempe et al.~\cite{kempe2003maximizing} introduced the Influence Maximization (IM) problem, which identifies $k$ most influential nodes (spreader nodes) such that the maximum number of people in the network adopt the promotional message. 
		
        In this work, we propose a maximization version of PAP called the \IM{} problem. Different people have different levels of interest in a particular product. This is modeled by assigning an interest value to each node in the network. Then, the problem is to select $k$ initial spreaders such that the sum of the interest values of the people (nodes) who become aware of the message is maximized. 
		
		We study the \IM{} problem under two popular diffusion models: the Linear Threshold Model (LTM) and the Independent Cascade Model (ICM). We show that the \IM{} problem is NP-Hard under LTM. We give linear programming formulation for the problem under LTM. We propose four heuristic algorithms for the \IM{} problem: \LBE{} (\LB{}), Maximum Degree First Heuristic (\MD{}), \PBE{} (\PB{}), and Maximum Profit Based Greedy Heuristic (\MP{}). Extensive experimentation has been carried out on many real-world benchmark data sets for both diffusion models. The results show that among the proposed heuristics, \MP{} performs better in maximizing the interest value. }

	
	
	
	\maketitle
	
	\section{Introduction}\label{sec:intro}
	Due to the increasing use of smartphones, people are connected to their friends, family, or customers through the internet. 
	We call the network of people a social network (SN). People on social media look at the information sent by their friends and either use the information for their interest, forward it to their friends, or do both. In this way, the information propagates in the network. Social networks become very important for marketing, political campaigns, and promoting products through e-commerce platforms.  
	
	E-commerce business is growing very fast across the world in urban as well as rural areas~\cite{rural2020}. 
	Small businesses use social media to grow and compete by advertising their products on social media. Social media networks help companies to attract customers without much physical effort. 
	When someone gets information from an important person, they start believing in the information and may forward the same to their friends. Hence, companies select a few highly influential people on social media to advertise their products. Identifying influential people in social networks is known as the Influence Maximization (IM) problem. 

	The diffusion model is the process by which information propagates in social networks. There are two fundamental diffusion models: the Linear Threshold Model (LTM)  and the Independent Cascade Model (ICM). \citeauthor{kempe2003maximizing}~\cite{kempe2003maximizing} propose the Influence Maximization (IM) problem in social networks. The IM problem is also known as the Target Set Selection problem. Two versions of the IM problem exist under the LTM. Maximization version: The input for the IM problem is a graph $G$ and a positive integer $k$. The objective is to find a set of at most $k$ highly influential nodes $S \subseteq V$ that maximize the number of influenced nodes in the network. In the minimization variant, for a given graph $G$ and a positive integer $l$ where $l \le |V|$, we need to find a seed set $S \subseteq V$ with minimum cardinality that influences at least $l$ vertices.
	
	\citeauthor{cordasco2018evangelism}~\cite{cordasco2018evangelism} propose a variant of the minimization version of the IM problem under the LTM where $l=|V|$. The problem is the Perfect Evangelising Set (PES) in social networks. The PES problem has two thresholds for each vertex: influence threshold $t_I$ and activation threshold $t_A$. The influence threshold is always less than or equal to the activation threshold. A vertex has three states: {\it non-aware}, {\it influenced}, and {\it activated (spreader)}. Initially, all the vertices of the graph are non-aware. We select some set of initial spreaders $S \subseteq V$, which are assumed to be active vertices. A vertex with sufficient active neighbors becomes influenced or activated when the vertex satisfies the respective influence or activation threshold. An influenced node is assumed to believe the information but does not forward the information. On the other hand, an activated node is an influenced node that forwards (spreads) the information to its neighbors. The objective is to find the minimum number of initial spreaders that influence all the nodes of the graph. Later, \citeauthor{cordasco2019active}~\cite{cordasco2019active} present the Perfect Awareness Problem (PAP), which is a specialization of the PES problem where $t_I(v)=1$, $\forall v \in V$. Here, as the influence threshold is $1$, the influenced nodes are also termed as {\it aware} nodes. 
	
	We propose a maximization version of the PAP problem, which we call the { \IM{}} problem.  The motivation for the problem is that a small company, compared to a blue-chip company, can compete by giving discounts on products and providing quality and indigenous products but may not be able to advertise in a big way due to financial constraints. The idea for small firms is to try to target highly interested buyers while advertising the product.  This strategy improves the chances of selling a product to companies with low resources. 
	
	Inputs for the \IM{} problem are a graph $G(V, E, t_A,\eta)$ and a positive integer $k$ denoting the size of the seed set,  where $V$ is the set of vertices, $E$ is the set of edges, $t_A: V \rightarrow  \mathbb{Z}^{+}$ is the activation threshold function and $\eta: V \rightarrow (0,1]$ is the interest function, and we know that a node having more interest value resists less in spreading of the information. 
	Note that, like in PAP, this problem also has $t_I(u)=1$ for all $u \in V$.  The goal is to find a seed set of size at most $k$ that maximizes the sum of interest values associated with all the influenced (aware) vertices. 
	Throughout the paper, we use the words influenced and aware synonymously;  activated and spreader are also synonymously used. 
	
	Initially, all the vertices of the graph $G$ are in a non-aware state. A vertex changes its state from non-aware to aware or aware to the spreader in only one direction. A seed set $S\subseteq V$ is a set of initially activated vertices, which are called the initial spreaders. The vertex $u \in S$ spreads information immediately to its neighbors. In the spreading process, a vertex $v \notin S$ can be a spreader with the condition $t_A(v)\le |N(v)\cap S|$ where $N(v)$ is a set of neighbors of the vertex $v$. 

	The spreading process stops when there is no change in the status of the number of aware/influenced nodes.  The \IM{} problem aims to find a set $S\subseteq V$ that maximizes the sum of the interest values of the influenced vertices. 
	
	

	
    The list of our contributions in this paper is as follows:
    \begin{enumerate}
        \item We propose the \IM{} problem.
        \item We first prove that under LTM, the \IM{} problem is NP-Hard.
        \item We provide an LP formulation for \IM{} under LTM. By using the Gurbi solver~\cite{gurobi}, we find optimal solutions for small datasets.
        \item We present four heuristic algorithms for \IM{}. 
         \begin{itemize}
             \item \LBE{}.
             \item  \MDE{}.
             \item \PBE{}.
             \item \MPBF{}.
         \end{itemize}
        \item The proposed heuristics are tested on real-world benchmark data sets for both the diffusion models LTM and ICM.
    \end{enumerate}
	
	The paper is organized as follows: The recent studies on information spreading in social networks are discussed in Section~\ref{sec:relwork}. The problem definition, NP-Hard reduction from max-coverage-problem to \IM{}, and LP formulation for \IM{} under the linear threshold model are discussed in Section \ref{sec:proposedProblems}. In Section~\ref{sec:proposedHeuristics}, we propose heuristics: \LBE{}, \MDE{}, \PBE{}, and \MPE{}. The heuristics are tested on real-world data sets. The analysis of the outcome of heuristics is discussed in Section~\ref{sec:resultanddiscussion}.  The paper is finalized with the conclusion in Section~\ref{sec:conclusion}.

	\section{Related Work}\label{sec:relwork}
	
	Information spreading is a trending area for research in this digital world, where people are connected through social media. \citeauthor{kempe2003maximizing}~\cite{kempe2003maximizing} proposed  Influence Maximization (IM) problem in social networks. Inputs for the IM problem are graph $G$ and $k\in \mathbb{Z}^+$. The objective is to select a set of initial spreaders $S\subseteq V(G)$ (seed set) of size $k$ that maximizes influenced people in the network. Another variant of IM is to find the minimum size of the seed set that influences the whole graph. \citeauthor{kempe2005influential}~\cite{kempe2005influential} also propose a greedy algorithm with an approximation factor $(1-1/e-\epsilon)$ for Decreasing Cascade Model. \citeauthor{chen2009approximability}~\cite{chen2009approximability} shows the Target Set Selection problem is hard to approximate less than the poly-logarithmic factor. \citeauthor{cordasco2018evangelism}~\cite{cordasco2018evangelism} propose a Perfect Evangelising Set (PES) in social networks where each vertex changes its state among three states: non-influenced (non-aware), influenced, and spreader. We need to find the minimum size of the seed set that influences the whole graph under the Linear Threshold Model. A similar problem is the Perfect Awareness Problem~\cite{cordasco2019active} (PAP). The difference between PAP and PES is that an influenced node must have sufficient spreader neighbors in PES. On the other hand, an influenced node must have at least one spreader in the PAP problem. For the Perfect Awareness Problem, ~\citeauthor{cordasco2019active}~\cite{cordasco2019active} propose an exact algorithm for trees and a heuristic for the general graphs. Recently, \citeauthor{pereira2021effective}~\cite{pereira2021effective} and ~\citeauthor{2023RahulCent}~\cite{2023RahulCent} proposed heuristics which improve results for general graphs. In real-life scenarios, information originating from a vertex does not continuously spread. The information may spread up to some hops from the source of the information. Recently,~\citeauthor{2023tieredInf}~\cite{2023tieredInf}  addressed this issue and introduced a variant of the Target Set Selection problem called target set selection in social networks with tiered influence and activation thresholds.
	
	Opinion maximization is a variant of the IM problem.  In the opinion maximization (OM) problem, initially, all vertices are inactive. The task is to pick highly influential people who maximize the sum of people's opinions~\cite{openion2013}. The difference between OM and \IM{} is that in OM, a vertex has two states, active and inactive, but in \IM{}, a vertex has three states: non-aware,  aware, and spreader. In the \IM{}, the interest value of the vertices is part of the input, but opinion is calculated for each vertex. In other words, we can say that the interest values of the vertices are independent of each other. \citeauthor{alla2023opinion}~\cite{alla2023opinion} proposed heuristics for the OM problem based on centrality measures and clustering. For a given graph $G$ and integer $k$, finding the maximum number of influential nodes is studied by \citeauthor{inf-martingale}~\cite{inf-martingale}. 
	
	Some of the related problems to information spreading are Graph Burning~\cite{gautam2022faster}, $k$-center~\cite{k-center}, the Target Influence Maximization problem in competitive networks (TIMC) based on the Independent Cascade Model proposed by~\citeauthor{2023targetedCompt} in~\cite{2023targetedCompt}, and Rumor Minimization~\cite{rumor2020containment}. 
	Based on centrality measures, \citeauthor{gautam2022faster}~\cite{gautam2022faster} propose three heuristics for the Graph Burning problem, and very recently, \citeauthor{2023burningNumber}~\cite{2023burningNumber} give a genetic algorithm for Graph Burning based on centrality measure. In the $k$-center problem, we need to establish $k$ warehouses that minimize the maximum distance from people to warehouses. Rumor Minimization is just stopping rumors by spreading truths among rumor-adopted vertices in the networks. 
	
	We study \IM{} problem under the Linear Threshold Model (LTM) and Independent Cascade Model (ICM). Under LTM, we prove that the problem is NP-Hard, and we provide a linear programming formulation for the problem. We propose four heuristics for the \IM{} problem. The heuristics are tested on real-world datasets under the diffusion models LTM and ICM.  
	
	\section{Interest Maximization}
	Under the Linear Threshold Model (LTM), we show that the \IM{} problem is NP-Hard and propose the linear programming formulation. We also find optimal solutions by solving Linear Equations using Gurbi-Solver~\cite{gurobi}. 

	\label{sec:proposedProblems}
	\subsection{Interest Maximization under LTM }
	An LTM generally has the following parameters: A weighted directed graph $G = (V,E)$, vertex threshold values $ t_A(u) \in \mathbb{Z^+}$ for all $u \in V(G)$. In the \IM{} problem under LTM, apart from the LTM parameters for each vertex, we have the interest value $\eta(u)$, where $0 < \eta(u)\le 1$. The value $\eta(u)$ shows how much the vertex $u$ is interested in the product's advertisement.
	
	The diffusion process for the \IM{} under LTM~\cite{cordasco2019active}  is as  follows:
	\begin{enumerate}
		\item Initially, all vertices are non-aware, i.e., the aware set $A=\phi$.
		\item Select a set of vertices $S$ as initial spreaders to start spreading the information.
		\item A non-aware vertex gets aware when it is a neighbor of at least one spreader vertex. When a vertex gets aware, it is added to the aware set $A$.
		\item A non-spreader vertex $u$ becomes a spreader when it is a neighbor of at least $t_A(u)$ number of spreaders.  
		\item The diffusion process is repeated until no more vertices change their state from non-spreader to spreader.
	\end{enumerate}
	
	The objective is to find a seed set $S$ of size $k$ that maximizes the sum of the interest values $I = \sum_{u \in A} \eta(u)$ under the Linear Threshold Model, where $A$ is the final set of aware (influenced) vertices for the seed set $S$.
	
	\subsubsection{Interest Maximization is NP-hard} \hfill\\
	\hfill\\
	We reduce the decision version of the Maximum Coverage Problem (MCP) to the \IM{} problem.
	The decision version of the Maximum Coverage Problem is as follows:\\

	\noindent\fbox{%
		\parbox{\textwidth}{%
			\textbf{Input: }  Universe of elements $U$, $m$ subsets $S = \{S_1,S_2,S_3 \cdots S_m \}$, and $k,l \in \mathbb{Z+}$.\\
			\textbf{Question: } Are there   $k$ subsets that cover at least $l$ elements? \\
		}%
	}

	For a given instance of the Maximum Coverage Problem (MCP), we construct an instance of the \IM{} problem as follows:
	\begin{enumerate}
		\item For each element $u_i \in U$, introduce a vertex labeled $u_i$.
		\item For each subset $S_i$, introduce a vertex labeled $S_i$.
		\item If $u_i\in S_j$ add  an edge between the vertices labeled $u_i$ and $S_j$.
	\end{enumerate}
	Note that the constructed graph $G'$ is bipartite. We set the influence and interest values of all the vertices of $G'$ to one. The activation threshold is set as $t_A(u)=deg(u)$ and $\eta(u)=1$ for all $u \in  V(G')$.

	
	\begin{lemma}
		There exist $k$ subsets that cover at least $l$ elements if and only if there is a seed set of size $k$ that influences at least $k+l$ vertices.
	\end{lemma}
	\begin{proof}
		In the forward direction, if there are  $k$ subsets that cover at least $l$ elements, the vertices in $G'$ corresponding to the selected $k$ sets will be chosen as seed set in $G'$. As all these $k$ vertices influence at least $l$ vertices, 
		we can say that a seed set of size $k$ influences at least $k+l$ elements.
		
		In the backward direction, suppose we have a seed set $S$ containing elements of type  $S_i$ and $u_i$ that influence at least $k+l$ vertices. But our goal is to obtain a seed set having vertices of only $S_i$ type. Let the seed set $S$ contain some vertices $S_i$ from $S$  and vertices  $u_i$ in $U$ as shown in the \figurename~\ref{fig:np-hard}. For example, seed set is $S = \{S_1,S_3,S_4,u_2,u_5\}$. A vertex $u_i \in S$ can be swapped with one of its neighbors which are of $S_i$ type as shown in \figurename~\ref{fig:np-hard}. So, all vertices $u_i$ type can be swapped with neighboring $S_j$ type vertices, and it does not affect the optimality of the solution because the graph is bipartite. The seed set $S$ contains only the top vertices of the bipartite graph, as shown in \figurename~\ref{fig:np-hard}, and influences at least $k+l$ vertices. So, the seed set $S$ is the solution having $k$ subsets that cover at least $l$ elements.
	\end{proof}
	
	\subsubsection{LP-formulation } \hfill\\
	\hfill\\
	We propose the LP formulation of the \IM{} problem under the linear threshold model. $G(V, E, t_A, \eta)$ is a given graph, where $V$ denotes the set of vertices, $E$ is the set of edges, $t_A(u)$ is the activation threshold value of vertex $u$, and $\eta(u)$ is the interest value of vertex $u$. For all $u \in V$, if $u$ is influenced in at most $r$ rounds, $I_{u,r} =1$; otherwise, $I_{u,r} = 0$ . Similarly, for all $u \in V$, if $u$ becomes a spreader by $r$ rounds, then  $A_{u,r} = 1$;  otherwise, $A_{u,r} = 0$. 
	\begin{align}
		Objective: \text{Maximize}  \displaystyle\sum\limits_{ u\in V } I_{u,n}*\eta(u) & & & \label{eq:1}\\
		\text{subject to the constraints} \sum\limits_{u \in V} A_{u,0}\le k ,  & &  \label{eq:2}\\
		t_A(u)*(A_{u,r}-A_{u,0}) \leq \sum\limits_{v \in N(u) }  A_{v,r-1},  & &\forall u \in V , r\in [1,n] \label{eq:3}\\
		I_{v,n} \leq (\sum\limits_{u \in N[v] } A_{u,n} + A[v,0] ) *I_{v,n}, & & \forall v \in V \label{eq:4} \\
  \sum\limits_{u \in N[v] } A_{u,n} + A[v,0] \leq (\sum\limits_{u \in N[v] } A_{u,n} + A[v,0] ) *I_{v,n}, & & \forall v \in V\label{eq:5}
	\end{align}
	
	The objective is to maximize the objective function  (\ref{eq:1}).  
	Equation (\ref{eq:2}) forces to select at most $k$ vertices as initial spreaders. 
	As and when the vertex $u$ has $t_A(u)$ number of spreader neighbors, Equation (\ref{eq:3}) forces  $A_{u,r}$ to be $1$ except for initial spreader vertices $A_{u,0}$, thus making $u$ a spreader node in the $r^{th}$ iteration where $r \in [1,n]$.  If at least a neighbor of $v$ is active till $n$ rounds, then equations (\ref{eq:4}) and (\ref{eq:4}) forces the $I_{v,n}=1$ (Equation (\ref{eq:4}) forces $I_{v,n}=1$ if $\sum\limits_{u \in N[v] } A_{u,n} + A[v,0]$ is equal to $1$. But this equation also satisfies when $\sum\limits_{u \in N[v] } A_{u,n} + A[v,0] =1$ and  $I_{v,n}=0$ which is incorrect. Equation (\ref{eq:5}) supports this condition and forces $I_{v,n}=1$).

    We use the Gurbi~\cite{gurobi} package to solve Linear Programming Equations. The results are as follows. For the dataset  Karate~\ref{tab:graphs}:
    \begin{enumerate}
        \item If $k=1$, then the sum of interest is $11.22$.
        \item If $k=2$, sum of interest is $14.47$.
        \item for $k=3$, sum of interest is $15.94$ and all vertices get influenced.
    \end{enumerate}
    For the dataset  Jazz~\ref{tab:graphs}:
    \begin{enumerate}
        \item If $k=5$, then the sum of interest is $86.64$.
        \item If $k=10$, sum of interest is $96.64$.
        \item for $k=12$, sum of interest is $97.44$.
        \item for $k=13$, sum of interest is $97.46$ and all vertices get influenced.
    \end{enumerate}

	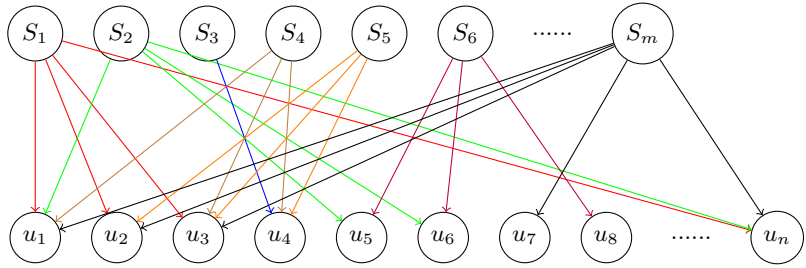
\begin{figure}[!htbp]
		\centering
		\begin{tikzpicture}[
			node distance =20mm and 4mm,
			C/.style = {circle, draw, minimum size=1em},
			S/.style = {circle,fill=black!25,draw, minimum size=1em},
			every edge quotes/.style = {auto, font=\footnotesize, sloped}
			] 
			
			\node (0) [C, ] {$S_1$};
			\node (1) [C,  right = of 0] {$S_2$};
			\node (2) [C,  right=of 1] {$S_3$};
			\node (3) [C,   right=of 2] {$S_4$};
			\node (4) [C,   right=of 3] {$S_5$};
			\node (5) [C,  right= of 4] {$S_6$};
			\node (16)[rectangle,  right= of 5] {......};
			\node (6) [C,   right=of 16] {$S_{m}$};
			\node (7) [C,  below =of 0] {$u_1$};
			\node (8) [C,   right=of 7] {$u_2$};
			\node (9) [C,   right=of 8] {$u_3$};
			\node (10) [C,  right=of 9] {$u_4$};
			\node (11) [C,  right=of 10] {$u_5$};
			\node (12) [C,  right=of 11] {$u_6$};
			\node (13) [C,  right=of 12] {$u_7$};
			\node (14) [C,  right=of 13]{$u_8$};
			\node (17)[rectangle,  right= of 14] {......};
			\node (15) [C,   right=of 17] {$u_n$};
			\foreach \x in {0}
			\foreach \y in {7,8,9,15}
			{
				\draw[red] (\x)--(\y) ;
			}
			
			\foreach \x in {1}
			\foreach \y in {7,12,15,11}
			{
				\draw[green](\x)--(\y) ;
			}
			\foreach \x in {2}
			\foreach \y in {10}
			{
				\draw[blue](\x)--(\y) ;
			}
			
			\foreach \x in {3}
			\foreach \y in {7,9,10}
			{
				\draw[brown](\x)--(\y) ;
			}
			
			\foreach \x in {4}
			\foreach \y in {10,9,8}
			{
				\draw[orange](\x)--(\y) ;
			}
			
			\foreach \x in {5}
			\foreach \y in {14,11,12}
			{
				\draw[purple](\x)--(\y) ;
			}
			\foreach \x in {6}
			\foreach \y in {7,8,9,15,13}
			{
				\draw[](\x)--(\y) ;
			}
			
		\end{tikzpicture}
		\caption{The sets $S_1,S_2,S_3,S_4,S_5,S_6,$ and $S_m$ cover elements $u_1,u_2,,u_3,\cdots$ and $u_n$. The figure shows that each set $S_j$($1\le j \le m)$ covers neighboring vertices (depicted by edges of the same color). }
		\label{fig:np-hard}
	\end{figure}

	
	
	\subsection{Interest Maximization under the ICM}
	We have a directed graph $G(V, E)$ with weights on edges, and interest value $\eta(u)$ is associated with each node $u$. A vertex $u$ activates vertex $v$ with probability $p(u,v)$, where $p(u,v)$ is the weight on edge $(u,v)$. Once a vertex $v$ becomes active (spreader), $v$ can try to activate its neighbors once, and the spreader vertex does not change its state again. 
	
	The diffusion under the ICM model is as follows:
	\begin{itemize}
		\item Initially, all vertices $u \in V$ are in an inactive state (non-spreader vertex).
		\item Select seed set $S \subseteq V$ of given size $k$, and all the vertices in $S$ are assumed to be in active state. Initialize the set of active vertices $A$ as  $A = S$.
		\item An active vertex $u$ activates its out-neighbors $v$ with probability $p(u,v)$. Once a vertex $u$ changes the state from inactive to active, add $u$ to $A$. A vertex can activate another vertex only once.
		\item The diffusion process stops when no more vertices change their state from inactive to active.
	\end{itemize}
	
	The objective is to find a seed set $S$ of size $k$ that maximizes the interest value  $\sum_{u \in A} \eta(u)$,  where $A$ is the set of final active vertices. 

 \textbf{We propose four heuristics and their time complexity analysis for \IM{} problem are discussed in the next section.} 
	
	\section{Proposed Heuristics}
	\label{sec:proposedHeuristics}
	We discuss four heuristics in this section. The level-based greedy and degree-based heuristics select initial spreaders based on the maximum number of neighbors and degrees, respectively. The rest of the heuristics select the spreader vertex with the highest profit iteratively. The heuristic algorithms are discussed below.
	
	\subsection{Level Based Greedy Heuristic (\LB)}
	\label{algo:h-Level-wise}
	
	\begin{figure}
		\centering
		\includegraphics[scale=0.1]{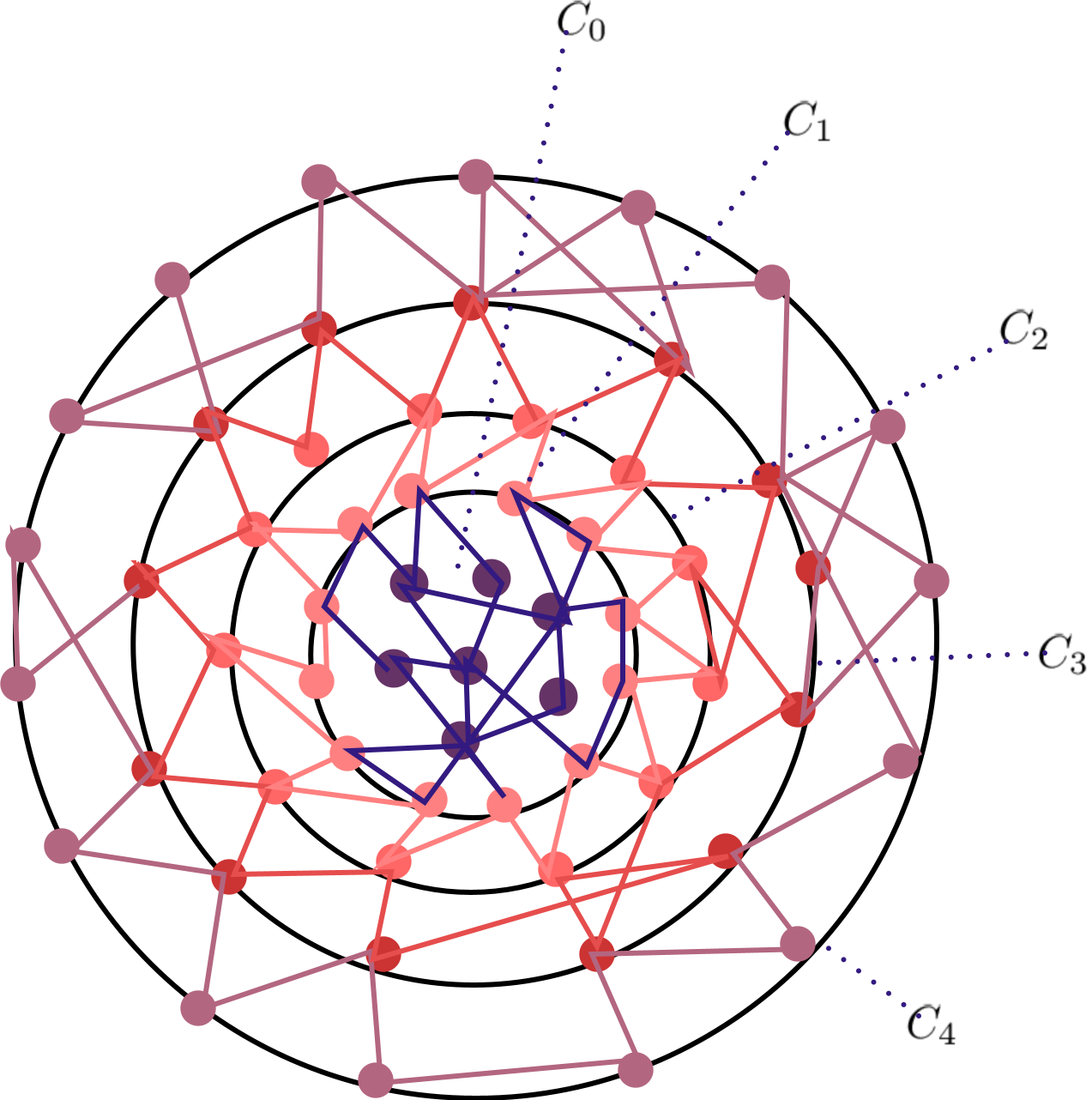}
		\caption{Level-wise, select highly influenced vertices from the graph.}
		\label{fig:level-wise}
	\end{figure}
	The \LBE{} (\LB{}) in turn calls three procedures: the \LBIF{}, \spr{} and \df{}.
	\LBIF{} procedure (as shown in \algorithmcfname~\ref{algo:levelBasedinfluence}) computes and returns a specific list $L^*$, which is a permutation of the vertex set $V(G)$. The permutation is generated by considering the degree and the interest values of the vertices as described below. 
	
	Let $L_0=L[0,1,2,3\cdots n/2]$ be the first half of the vertices of the list $L$, where $L$ is the sorted list of the vertex set $V(G)$ in decreasing order of interest values. We consider $L_0$ to be core vertices that are at level zero. As shown in \figurename~\ref{fig:level-wise-example}, we do a level order traversal of the graph $G$ based on the distance of the vertices from level zero.  For $1\le i \le e(L_0)$, $L_i= N(L_i) \backslash \cup_{j=i-1}^{j=0} L_j$, where $N(L_i)$ is the set of the open neighborhood of the vertices in $L_i$ and $e(L_0)$ is the maximum eccentricity of the vertices in the list $L_0$. The leveling of the graph $G$ continues till all the vertices of the graph $G$ are covered. These levels of $G$ are depicted in \figurename\ref{fig:level-wise} as concentric circles with the innermost circle representing level $L_1$ and center with core vertices in list $L_0$.
	
	Now, we sort the list of $L_i$ in decreasing order according to the degree of the vertices. For each index $j = 0,1,2,3 \cdots $, we pick the vertices located at the index $j$ of each list $L_i$, and these vertices are sorted in decreasing order of the degree and appended to the list $L^*$. The same process is repeated for each index $i=1,2,3 \cdots $. As the lists may not be of  $j+1$ length, only if the vertex is available at an index $j$ the vertex is picked. This process aims to give equal importance to high-degree vertices at each level.

	The \spr{}  function, as given in the Algorithm~\ref{algo:spreader}, takes as input an ordered list of vertices $L^*$ and \GR{}. $A$ and $S$ are the aware and seed sets, respectively. Initially, the seed set and aware set are empty. For each vertex $u\in L^*$,  if \df{} results in an increase in the number of influenced vertices with the seed set $S\cup \{u\}$, add $u$ to $S$ and decrease the number of required initial spreaders or seeds $k$ by one. Otherwise, ignore vertex $u$. The \spr{} procedure returns a seed set $S$. 
	
	The \df{} function, as given in the Algorithm~\ref{algo:diffusion} under LTM, returns a set of aware nodes $A$. The inputs for the \df{} function are graph $\GR{}$ and the seed set $S$. The \df{} function diffuses the information using the LTM or ICM and returns the set of aware nodes. 
	
	\textbf{Time Complexity Analysis:} The diffusion function takes time $O(|V|+ |E|)$, and the time complexity of sorting the vertices in decreasing order is $O(|V| \log(|V|))$. To find the eccentricity of the vertices in $L_0$ is computed in time at most $O(|V|+|E|)$ because we marge all vertices of $L_0$ into a single vertex and run a single breath first search. The leveling visits each vertex at once from the source vertices in $L_0$. It is equivalent to running a breadth-first search algorithm.  The rest of the algorithm sorts and prepares a list $L^*$ which is at most $O(|V| + |E|)$. The spreader function prepares the seed set by calling the diffusion function at most $|V|$ times. So, the time complexity of the heuristic is $O (|V|*(|V|+|E|))$.

	\begin{figure}
		\centering
		\begin{tikzpicture}[
			node distance =10mm and 5mm,
			C/.style = {circle, draw, minimum size=1em},
			S/.style = {circle,fill=black!25,draw, minimum size=1em},
			every edge quotes/.style = {auto, font=\footnotesize, sloped}
			] 
			
			\node (1) [C, fill=red!35, label=1] {1};
			\node (0) [C,fill=red!35, label=1, above= of 1] {0};
			\node (2) [C,fill=red!35, label=0.9, above right=of 1] {2};
			\node (3) [C, fill=red!35,label=0.7, below right=of 2] {3};
			\node (4) [C, fill=red!35,label=0.6, above right=of 2] {4};
			\node (5) [C,fill=red!35, label=0.6, right=of 2] {5};
			\node (6) [C, fill=red!35,label=0.6, below right=of 3] {6};
			\node (7) [C,fill=red!35,  label=0.5, below right=of 1] {7};
			\node (8) [C, fill=yellow!50,label=0.5, above right=of 5] {8};
			\node (9) [C, fill=yellow!50,label=0.5, right=of 5] {9};
			\node (10) [C, fill=yellow!50,label=0.4, right=of 3] {10};
			\node (11) [C, fill=green!50,label=0.3, right=of 9] {11};
			\node (12) [C, fill=yellow!50,label=0.3, right=of 10] {12};
			\node (13) [C, fill=green!50,label=0.3, below right=of 10] {13};
			\node (14) [C, fill=green!50,label=0.2, above right=of 11] {14};
			\node (15) [C, fill=green!50,label=0.2,  right=of 12] {15};
			\foreach \x in {0}
			\foreach \y in {2,4}
			{
				\draw (\x)--(\y) ;
			}
			
			\foreach \x in {1}
			\foreach \y in {2,5,3,7}
			{
				\draw (\x)--(\y) ;
			}
			\foreach \x in {2}
			\foreach \y in {3,5,7,4}
			{
				\draw (\x)--(\y) ;
			}
			
			\foreach \x in {3}
			\foreach \y in {7,9,10,6}
			{
				\draw (\x)--(\y) ;
			}
			
			\foreach \x in {5}
			\foreach \y in {10,9,4,8}
			{
				\draw (\x)--(\y) ;
			}
			
			\foreach \x in {9}
			\foreach \y in {14,11,12}
			{
				\draw (\x)--(\y) ;
			}
			\foreach \x in {12}
			\foreach \y in {15,13,6}
			{
				\draw (\x)--(\y) ;
			}
			
		\end{tikzpicture}

		\caption{ The vertices at level-0 $L_0=[0,1,2,3,4,5,6,7]$, at level-1 $L_1= [8,9,10,12]$, and at level-2 $L_2 = [11,13,14,15]$.}
		
	\label{fig:level-wise-example}
\end{figure}
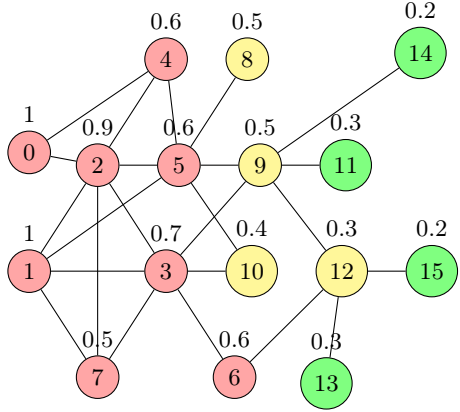

\begin{algorithm}
	\DontPrintSemicolon
	\SetKwInOut{Input}{Input}\SetKwInOut{Output}{Output}
	\Input   {Graph $\GR{} $ and Seed set $S$.}
	\Output  {The aware or influenced set $A$.}
	\SetKwFunction{FMain}{DIFFUSION}
	\SetKwProg{Fn}{Function}{}
	\Fn{\FMain{$G,S$}} {
		\Begin{
			$Q \gets [\; ]$\;
			$A \gets \phi$\;
			\For{$v \in S $}{
				$append(Q, v)$\;
			}
			\While {$empty(Q)=False$}
			{
				$v = removeFirst(Q)$\;
				\For {$w \in N[v]$}{
					\tcp{ $N[v] $ is the set of all neighbors of $v$ including $v$.}
					\If{ $ w \notin A $}{
						$A \gets A \cup \{w\}$\;
					}
					$S' \gets \{ x | x \in N[w] \cap S \}$\;
					\If {$t_A(w)\le |S'| $ and $ w \notin S$}
					{
						$S \gets S  \cup \{w\}$\;
						\tcp{ $w$ becomes the spreader.}
						$append(Q, w)$\;
					}
				}
			} 
			\KwRet{$A$}
		}		
	}
	\caption{Diffusion Function under LTM.}
	\label{algo:diffusion}
\end{algorithm}

\begin{algorithm}
	\DontPrintSemicolon
	\SetKwInOut{Input}{Input}\SetKwInOut{Output}{Output}
	\Input   {$\GR{} $.}
	\Output  {The seed set $S$.}
	\SetKwFunction{FMain}{Spreader}
	\SetKwProg{Fn}{Function}{}
	\Fn{\FMain{$G, L,k $}} {
		\Begin{
			$S\gets \phi$\;
			$A \gets \phi$\;
			$A \gets \text{DIFFUSION}(G,S)$\;
			\For{$u \in L$}{
				\tcp{ DIFFUSION function returns the set of aware vertices $A$ with the seed set $S \cup \{u\}$.}
				$A' \gets \text{DIFFUSION}(G,(S\cup \{u\}))$\;
				
				\If{$|A'|>|A|$}{
					${S} \gets {S} \cup \{u\}$\;
					$A \gets A'$\;
					$k \gets k -1$\;
				}
				\If{$k<0$}{
					\KwRet{${S}$}
				}
			}
		}
	}
	\caption{Spreader function finds the seed set $S$.}
	\label{algo:spreader}
\end{algorithm}

\begin{algorithm}
	\DontPrintSemicolon
	\SetKwInOut{Input}{Input}\SetKwInOut{Output}{Output}
	\Input   {Graph $\GR{}$}
	\Output  {A list of vertices in decreasing order of influential strength.}
	\SetKwFunction{FMain}{\LBIF{}}
	\SetKwProg{Fn}{Function}{}
	\Fn{\FMain{$G$}} {
		\Begin{
			$L \gets \text{list of vertices of V(G) in decreasing order of interest values}$\;
			$L_0 \gets L[0,\cdots n/2]$\;
			$L_1 \gets L[n/2+1,\cdots n]$\;
			\For{$i=1$ to $e(L_0)$}
			{
				\tcp{ $e(L_0)$ is the eccentricity of the vertices in $L_0$ by merging all vertices into the single vertex.}
				$L_i \gets  N(L_i) \backslash \cup_{l=i-1}^{l=0} L_l$\;
				$L_i \gets sort\_by\_degree\_dec(G,L_i)$\;
				
			}
			$l \gets max(\{|L_i| : i \in [0,e(L_0)] \; \text{ and } i \in \mathbb{Z}^+\})$\;
			$L^* \gets [\;]$\;
			\For{$j=0$ to $j = l-1$}{
				$temp \gets \phi$\;
				\For{$i=0$ to $i = e(L_0)$}{
					\If{$j< |L_i|$}{
						$temp \gets temp \cup L_i[j]$\;
					}
				}
				$temp \gets sort\_by\_degree\_dec(G,temp)$\;
				\ForEach{$u \in temp$}{
					$L^*.append(u)$
				}
				
			}

			\KwRet{$L^*$}
		}	
		
	}
	
	\caption{Level-Based-Influence.}
	\label{algo:levelBasedinfluence}
\end{algorithm}

\begin{algorithm}
	\DontPrintSemicolon
	\SetKwInOut{Input}{Input}\SetKwInOut{Output}{Output}
	\Input   {$\GR{}$, and the positive integer $k\in \mathbb{Z}^+$.}
	\Output  {The sum of interest value associated with influenced vertices with seed set $S$.}
	\SetKwFunction{FMain}{Level-Based-Greedy}
	\SetKwProg{Fn}{Function}{}
	\Fn{\FMain{$G,k$}} {
		\Begin{
			$L^* \gets \textbf{\LBIF{}(G)}$\;
			$S\gets \textbf{\spr{}(G,$L^*$)}$\;
			$A \gets \textbf{\df{}(G,S)}$\;
			\KwRet{$S,\sum_{u\in A} \eta(u)$}   
		}
	}
	\caption{\LBE{}.}
	\label{algo:level-wise-profit-main}
\end{algorithm}

\subsection{Maximum Degree First Heuristic (\MD{})}
\label{algo:h-max-degree-first}
A vertex having a higher degree can activate or influence more vertices. So, we sort the vertices in decreasing order of degree. Let the sorted array be $L$. Initially, all the vertices of the graph are non-aware. The aware set $A$ and the seed set $S$ are empty. For each vertex $u \in L$, if the \df{} function influences more additional vertices with seed set $S \cup \{u\}$, then add $u$ to $S$. Otherwise, ignore the vertex $u$. Once the size of the seed set reaches $k$, the process stops, and the set of influenced nodes $A$ for seed set $S$ is recorded. The Max-Degree First function returns $\sum_{u\in A} \eta(u)$.     

\textbf{Time Complexity Analysis:} The two major tasks of \MD{} are sorting the vertices in decreasing order which takes ($O(|V| \log(|V|)$) time and preparation of the seed set $S$ takes ($O(|V|*(|V|+|E|))$) time. So, the time complexity of the heuristic is  $O(|V|*(|V|+|E|))$.

\subsection{Profit Based Greedy Heuristic (\PB{})}
\label{algo:h-profit-based}
In this approach, before selecting the seed node $u$, compute the profit of the vertex $u$.  The procedure for computing the profit of the vertices is:  $N[u]$ is the set of closed neighbors of the vertex $u$, and the profit is $\sum_{v\in N[u]\backslash A} \eta(v)$, where $A$ is the current set of influence nodes. The non-activated vertex $v_p \in V$ with the maximum profit is added to the seed set $S$. 

As given in the Algorithm~\ref{algo:Profit-based-greedy}, the input parameters are graph \GR{} and positive integer $k$. The objective is to select a  $k$-size seed set $S$ that maximizes the sum of interest of influenced vertices. Initially, all the vertices of graph $G$ are in the non-aware state. The aware set $A$ and the seed set $S$ are empty. The algorithm iterates $k$ times for finding $k$ seeds. In each iteration, the profit on each vertex $u$ is calculated as $p(u) = \sum_{ v \in N[u]  \backslash A} \eta(v)$, where $p(u)$ is the profit on vertex $u$, and $N[u]$ is the set of closed neighbors of the vertex $u$. The maximum profitable vertex is selected as the seed node and added to the seed set $S$. The \df{} function diffuses the information with the seed set $S$ and appends all the influenced or aware vertices to $A$. We repeat the above steps up to $k$ times and find the seed set of size $k$. The \PBE{} returns the seed set as well as the sum of interest value associated with influenced vertices as $\sum_{u\in A} \eta(u)$. 

\textbf{Time Complexity Analysis:}  For each vertex $u \in V$, calculate the profit, find the maximum profitable vertex, and run diffusion. The profit formula is  $\forall u \in V$, $p(u)= \sum _{u \in N(u)\backslash A} \eta(u)$.  Calculation of the profit of vertices takes time at most $O(|V|*\Delta(G))$, where $\Delta(G)$ is the maximum degree of the graph $G$. These three tasks are completed in time $(|V|*\Delta(G))$, $(|V|)$, and $(|V| + |E|)$. The total time consumed by the algorithm is $O(|V| * (|V|+|E|))$.     




\begin{algorithm}
	\DontPrintSemicolon
	\SetKwInOut{Input}{Input}\SetKwInOut{Output}{Output}
	\Input   {$\GR{} $ and positive integer $k\in \mathbb{Z}^+$.}
	\Output  {The seed set $S$ and the sum of interest  of  aware/influenced vertices.}
	\SetKwFunction{FMain}{\PBGF{}}
	\SetKwProg{Fn}{Function}{}
	\Fn{\FMain{$G,k $}} {
		\Begin{
			$S\gets \phi$\;
			$A \gets \phi$\;
			\For{ $i = 1$ to $k$}{
				$A \gets \df{}(G,S)$\;
				$v_{p}\gets -1$\;
				\For{$u \in V$}{
					\tcp{ DIFFUSION function returns set of aware vertices $A$ with seed set $S$.}
					$p(u) \gets \sum_{ v \in N(u) \text{ and } v \notin A}\eta(v) $\;
					
					\If{$m>p(u)$}{
						$v_p \gets u$\;
						$m  \gets p(u)$\;
					}
				}
				$S \gets S \cup \{ v_p \} $\;
			}
			$A \gets \df{}(G,S)$\;
			\KwRet{$S,\sum_{ u \in A }\eta(u)$}
		}
	}
	\caption{\PBE{}}
	\label{algo:Profit-based-greedy}
\end{algorithm}

	\subsection{Maximum Profit Based Greedy Heuristic (\MP{})}
	\label{algo:exp-five}
	As seen in the Algorithm~\ref{algo:Profit-based-greedy}, the highest profitable vertex is selected as a seed node in the Profit-Based Greedy Heuristic. The Profit-Based Greedy computes the incremental sum of interest in the neighborhood. The formula for the profit is $p(u) = \sum_{ v \in N[u] \backslash A} \eta(v)$ in the Profit Based Greedy Heuristic. $p(u)$ is the profit up to one hop. But we know when the seed set $S$ size is increased by adding a vertex $u$, the vertices that are more than one hop can be influenced. So, we extend the formula for the profit. For each $u \in V$, the \df{} function is run with the seed set $S \cup \{u \}$ to obtain the set of influenced vertices $A$ and $ p(u)= \sum_{ v \in A} \eta(v)$. The maximum profitable vertex $u_p$ is selected as the seed node and is added to $S$. The seed set $S$ of size $k$ is obtained by repeating the above process $k$ times. The \MPBF{} procedure returns a seed set and  $\sum_{ v \in A} \eta(v)$.
	
	\textbf{Time Complexity Analysis:} To select a seed vertex,  the time taken by the \MP{} is $O(|V| * (|V|+|E|))$ and moreover for $k$ seed vertices, the time complexity is $O(k* (|V| * (|V| + |E|)))$.
	

	
	
	\begin{algorithm}
		\DontPrintSemicolon
		\SetKwInOut{Input}{Input}\SetKwInOut{Output}{Output}
		\Input   {$\GR{} $ and positive inter $k$.}
		\Output  {The seed set $S$ and the sum of interest  of  aware/influenced vertices.}
		\SetKwFunction{FMain}{\MPBF{}}
		\SetKwProg{Fn}{Function}{}
		\Fn{\FMain{$G,k $}} {
			\Begin{
				$S\gets \phi$\;
				$A \gets \phi$\;
				\For{ $i = 1$ to $k$}{
					$v_p \gets -1$\;
					$m \gets 0$\;
					\For{$u \in V$}{
						\tcp{ \df{} function returns a set of aware vertices $A$ by seed set $S$.}
						$A \gets \text{\df{}}(G,S\cup\{u\})$\;
						$p(u) \gets \sum_{ v \in A}\eta(v) $\;
						
						\If{$m<p(u)$}{
							$v_p \gets u$\;
							$m  \gets p(u)$\;
						}
					}
					$S \gets S \cup \{ v_p \} $\;
				}
				$A \gets \df{}(G,S)$\;
				\KwRet{$S,\sum_{ u \in A }\eta(u)$}
			}
		}
		\caption{\MPE{}}
		\label{algo:max-profit-based}
	\end{algorithm}

	\textbf{In the next section, we analyze the performance of heuristics from the perspective of results and computation time of proposed heuristics.}
	
	\section{Results and Discussion}
	\label{sec:resultanddiscussion}
	We implement the \LBE{} (\LB{}),  \MDE{} (\MD{}),  \PBE{} (\PB{}), and \MPE{} (\MP{}) in the programming language Python. The system specifications are  Macbook Pro (2016) with 8GB RAM, 2.7Hz processor speed, and a hard disk space of 256GB.  The datasets obtained from the different sources are power~\cite{2015MarkNewMan}, BlogCatalog~\cite{reza2009social}, CA-HepTh~\cite{nr}, facebook~\cite{snapnets}, CA-GrQc~\cite{nr}, and CA-HepPh~\cite{nr}. We test all the heuristics on these data sets and tabulate the results.  
	
	The basic details of the data sets, along with their properties like the number of nodes, number of edges, density, average degree, and average clustering coefficient, are given in \tablename~\ref{tab:graphs}.
	
	We run all our heuristic algorithms with the diffusion models LTM and ICM. 
	The interest values $\eta$ of the nodes of the graph are generated randomly in the range $(0,1]$.

	\begin{table}[!h]
		\centering
		\caption{Data sets~\cite{2015MarkNewMan,snapnets,reza2009social,nr} and their network properties.}
		\scriptsize
		\begin{tabular}{p{2.5cm} p{1.5cm} p{1.5cm} p{1.5cm} p{2cm} p{1.5cm} p{1cm}}
			\hline
			Network & Nodes  & Edges & Density & Avg-Degree & Avg-CC \\
			\hline
			\hline
			power  &  4941  &  6594  &  0.0005   &  2.66  &  0.08 \\ \hline
			BlogCatalog  &  10312  &  333983  &  0.0063   &  64.77  &  0.46 \\ \hline
			CA-HepTh  &  9877  &  25998  &  0.0005   &  5.26  &  0.47 \\ \hline
			Karate  &  34  &  78  &  0.139  &  3.9706  &    0.571 \\ \hline
			facebook  &  4039  &  88234  &  0.0108   &  43.69  &  0.60 \\ \hline
			CA-GrQc  &  5242  &  14496  &  0.0011   &  5.53  &  0.53 \\ \hline
			CA-HepPh  &  12008  &  118521  &  0.0016   &  19.74  &  0.61 \\ \hline
			jazz  &  198  &  2742  &  0.1406  &  271.197  &  0.6 \\ \hline
		\end{tabular}
		\label{tab:graphs}
	\end{table}
	
	\subsection{Interest Maximization under the Linear Threshold Model }
    We compare our heuristics with ~\cite{2023RahulCent}. In CBH~\cite{2023RahulCent}, the seed vertex is picked based on centrality measures. However, the sum of the interests of neighbors may be higher than the centrality score of the vertex. So, \PB{} and \MP{} will outperform CBH~\cite{2023RahulCent}.
    
	The LTM is implemented with two thresholding mechanisms : $(i)$ $t_A(u) = \lceil deg(u)*0.5\rceil$, and $(ii)$ $t_A(u)= \lceil deg(u)*(1-\eta(u)) \rceil$, where $\eta(u)$ is the interest value of $u$ . The results of the proposed algorithms under $(i)$ are compared with \CB{}~\cite{2023RahulCent} by adapting the heuristic CBH to \IM{} problem, and the results are tabulated in \tablename~\ref{tab:threshold-without-related-interest} which are also plotted in \figurename~\ref{fig:threshold-without-related-interest}.
	Similarly, the results of the proposed algorithms under mechanism $(ii)$ are shown in the \tablename~\ref{tab:threshold-related-interest}, which are plotted in \figurename~\ref{fig:threshold-related-interest}. The

	\setlength\tabcolsep{1.2pt}
	\newcommand*\rot[1]{\rotatebox{90}{#1}}
	\renewcommand{\arraystretch}{1.25}

	\begin{table}[!htbp]
		
		\caption{The sum of interest values associated with influenced vertices with mechanism $t_A(u)=\lceil deg(u)*0.5\rceil$  (setting (i)) under the LTM model. }
		\centering
		\begin{tabular}{  c c c c c c c c c c c c } 
			\hline
			\multirow{2}{*} {Dataset} & \multirow{2}{*}{Algo.} & \multicolumn{10}{c}{Seed Set Size $10-100$}\\
			\cline{3-12}
			&  & \textbf{10} & \textbf{20} & \textbf{30} & \textbf{40} & \textbf{50} & \textbf{60} & \textbf{70} & \textbf{80} & \textbf{90} & \textbf{100} \\
			\hline
			
			\multirow{5}{*}{BlogCatalog}&   \CB{}~\cite{2023RahulCent} &  4358.3 &  4734.6 &  4852.0 &  4900.9 &  4961.6 &  5010.5 &  5029.5 &  5041.7 &  5071.9 &  5081.8\\
			\cline{2-12}
			&  \LB{} &  4058.0 &  4523.7 &  4715.5 &  4773.2 &  4849.0 &  4899.1 &  4919.5 &  4935.9 &  5020.4 &  5035.1\\
			\cline{2-12}
			&  \MD{} &  4288.3 &  4666.6 &  4769.1 &  4811.2 &  4848.7 &  4881.8 &  4921.3 &  5008.3 &  5024.2 &  5035.0\\
			\cline{2-12}
			&  \PB{} &  4356.4 &  4774.8 &  4923.7 &  4995.1 &  5034.5 &  5065.0 &  5086.1 &  5103.7 &  \textbf{5119.1} &  \textbf{5132.3}\\
			\cline{2-12}
			&  \MP{} &  \textbf{4356.4} &  \textbf{4774.8} &  \textbf{4923.7} &  \textbf{4995.1} &  \textbf{5034.8} &  \textbf{5065.0} &  \textbf{5086.6} &  \textbf{5104.7} &  \textbf{5119.7} &  \textbf{5132.7}\\
			
			\hline
			\multirow{5}{*}{CA-GrQc} &  \CB{}~\cite{2023RahulCent} &  183.8 &  306.4 &  406.6 &  500.5 &  577.3 &  636.9 &  687.0 &  743.6 &  795.3 &  848.6\\
			\cline{2-12}
			&  \LB{} &  89.5 &  178.1 &  265.6 &  343.8 &  433.8 &  501.4 &  571.7 &  637.1 &  687.2 &  737.2\\
			\cline{2-12}
			&  \MD{} &  105.8 &  156.0 &  240.7 &  309.5 &  409.8 &  504.9 &  567.3 &  654.0 &  737.1 &  788.5\\
			\cline{2-12}
			&  \PB{} &  221.3 &  363.0 &  482.2 &  582.0 &  669.1 &  745.6 &  816.8 &  881.5 &  945.3 &  999.0\\
			\cline{2-12}
			&  \MP{} &   \textbf{222.0} &  \textbf{366.9} &  \textbf{483.4} &  \textbf{585.8} &  \textbf{674.6} &  \textbf{754.2} &  \textbf{826.8} &  \textbf{893.5} &  \textbf{954.8} &  \textbf{1012.6}\\
			
			\hline
			\multirow{5}{*}{CA-HepPh}  &  \CB{}~\cite{2023RahulCent} &  599.6 &  938.1 &  1172.0 &  1330.4 &  1457.7 &  1605.1 &  1728.0 &  1804.8 &  1884.5 &  1952.7\\
			\cline{2-12}
			&  \LB{} &  536.4 &  681.6 &  781.6 &  859.0 &  903.2 &  947.8 &  1011.7 &  1096.3 &  1175.4 &  1198.8\\
			\cline{2-12}
			&  \MD{} &  499.7 &  633.0 &  730.5 &  789.4 &  835.0 &  883.8 &  925.1 &  941.5 &  970.2 &  1036.4\\
			\cline{2-12}
			&  \PB{} &  815.2 &  1147.2 &  1404.8 &  1604.3 &  1772.1 &  1922.1 &  2051.5 &  2168.1 &  2279.8 &  2381.3\\
			\cline{2-12}
			&  \MP{} &  \textbf{815.6 }&  \textbf{1147.2} &  \textbf{1404.8} &  \textbf{1605.5} &  \textbf{1775.3} &  \textbf{1922.3} &  \textbf{2054.6} &  \textbf{2173.2} &  \textbf{2283.9} &  \textbf{2387.3}\\
			
			\hline
			\multirow{5}{*}{CA-HepTh}  &  \CB{}~\cite{2023RahulCent} &  224.1 &  373.4 &  515.3 &  627.3 &  732.7 &  817.5 &  898.9 &  975.4 &  1061.7 &  1136.6\\
			\cline{2-12}
			&  \LB{} &  217.9 &  361.1 &  480.1 &  575.3 &  663.2 &  755.7 &  827.2 &  879.0 &  920.2 &  956.7\\
			\cline{2-12}
			&  \MD{} &  233.6 &  364.3 &  465.0 &  566.1 &  676.0 &  735.0 &  828.3 &  872.9 &  931.3 &  993.4\\
			\cline{2-12}
			&  \PB{} &  257.5 &  434.0 &  578.3 &  709.4 &  824.4 &  933.2 &  1035.1 &  1129.5 &  1220.6 &  1306.2\\
			\cline{2-12}
			&  \MP{} &  \textbf{257.5} &  \textbf{436.7} &  \textbf{582.2} &  \textbf{709.4} &  \textbf{826.1} &  \textbf{935.8} &  \textbf{1039.3} &  \textbf{1137.0 }&  \textbf{1229.1} &  \textbf{1315.2}\\
			
			\hline
			\multirow{5}{*}{facebook} &  \CB{}~\cite{2023RahulCent} &  \textbf{1987.5} &  \textbf{1987.5} &  \textbf{1987.5} &  \textbf{1987.5} &  \textbf{1987.5} &  \textbf{1987.5} &  \textbf{1987.5} &  \textbf{1987.5} &  \textbf{1987.5} &  \textbf{1987.5}\\
			\cline{2-12}
			&  \LB{} &  1883.3 &  \textbf{1987.5} &  \textbf{1987.5} &  \textbf{1987.5} &  \textbf{1987.5} &  \textbf{1987.5} &  \textbf{1987.5} &  \textbf{1987.5} &  \textbf{1987.5} &  \textbf{1987.5}\\
			\cline{2-12}
			&  \MD{} &  1938.2  &  \textbf{1987.5} &  \textbf{1987.5} &  \textbf{1987.5} &  \textbf{1987.5} &  \textbf{1987.5} &  \textbf{1987.5} &  \textbf{1987.5} &  \textbf{1987.5} &  \textbf{1987.5}\\
			\cline{2-12}
			&  \PB{} &  \textbf{1987.5} &  \textbf{1987.5} &  \textbf{1987.5} &  \textbf{1987.5} &  \textbf{1987.5} &  \textbf{1987.5} &  \textbf{1987.5} &  \textbf{1987.5} &  \textbf{1987.5} &  \textbf{1987.5}\\
			\cline{2-12}
			&  \MP{} &  \textbf{1987.5} &  \textbf{1987.5} &  \textbf{1987.5} &  \textbf{1987.5} &  \textbf{1987.5} &  \textbf{1987.5} &  \textbf{1987.5} &  \textbf{1987.5} &  \textbf{1987.5} &  \textbf{1987.5}\\
			
			\hline
			\multirow{5}{*}{power} &  \CB{}~\cite{2023RahulCent} &  90.5 &  158.8 &  223.4 &  285.9 &  350.0 &  406.4 &  470.9 &  523.4 &  569.5 &  612.5\\
			\cline{2-12}
			&  \LB{} &  84.3 &  155.8 &  209.1 &  260.6 &  304.7 &  353.5 &  413.7 &  462.2 &  505.7 &  552.8\\
			\cline{2-12}
			&  \MD{} &  96.5 &  160.0 &  229.3 &  291.4 &  342.9 &  395.5 &  451.2 &  505.2 &  538.2 &  606.5\\
			\cline{2-12}
			&  \PB{} &  106.3 &  182.0 &  247.1 &  325.9 &  387.4 &  449.4 &  503.6 &  557.3 &  608.8 &  663.4\\
			\cline{2-12}
			&  \MP{} &  \textbf{122.4} & \textbf{ 220.7} & \textbf{ 307.3} &  \textbf{384.9} &  \textbf{465.5} &  \textbf{542.1} &  \textbf{614.5} &  \textbf{687.9} &  \textbf{761.3} &  \textbf{825.6}\\
			\hline
		\end{tabular}
		\label{tab:threshold-without-related-interest}
	\end{table}
	\begin{figure}[!htbp]
		\centering
		\includegraphics[scale=0.30]{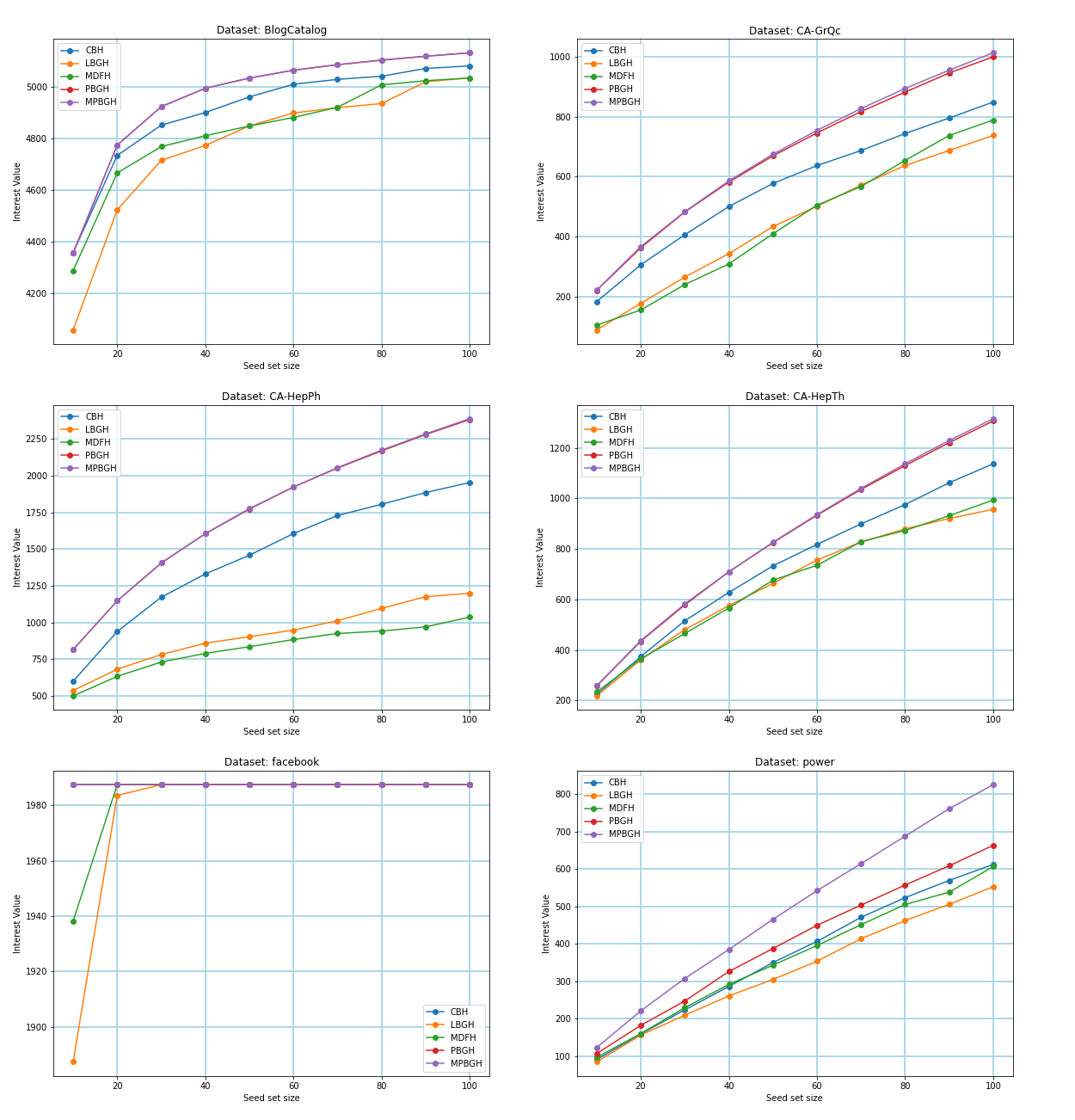}
		\caption{The sum of interest values associated with influenced vertices with mechanism $t_A(u)=\lceil deg(u)*0.5\rceil$ (setting (i)) under the LTM model. }
		\label{fig:threshold-without-related-interest}
	\end{figure}
	
	
	\begin{table}[!htbp]
		\caption{The sum of interest values associated with influenced vertices with mechanism $t_A(u)= \lceil deg(u)*(1-\eta(u)) \rceil$ (setting (ii)) under the LTM model.}
		\centering
		\scriptsize
		\begin{tabular}{  c c c c c c c c c c c c } 
			\hline
			\multirow{2}{*} {Dataset} & \multirow{2}{*}{Algo.} & \multicolumn{10}{c}{Seed Set Size $10-100$}\\
			\cline{3-12}
			&  & \textbf{10} & \textbf{20} & \textbf{30} & \textbf{40} & \textbf{50} & \textbf{60} & \textbf{70} & \textbf{80} & \textbf{90} & \textbf{100} \\
			\hline
			\multirow{5}{*}{\rot{BlogCatalog}}  &  \CB{}~\cite{2023RahulCent} &  5175.3 &  5189.8 &  \textbf{5193.7} &  \textbf{5193.7} &  \textbf{5193.7} &  \textbf{5193.7} &  \textbf{5193.7} &  \textbf{5193.7} &  \textbf{5193.7} &  \textbf{5193.7}\\
			\cline{2-12}
			&  \LB{} &  5175.8 &  5192.1 &  \textbf{5193.7} &  \textbf{5193.7} &  \textbf{5193.7} &  \textbf{5193.7} &  \textbf{5193.7} &  \textbf{5193.7} &  \textbf{5193.7} &  \textbf{5193.7}\\
			\cline{2-12}
			&  \MD{} &  5175.2 &  5190.4 &  \textbf{5193.7} &  \textbf{5193.7} &  \textbf{5193.7} &  \textbf{5193.7} &  \textbf{5193.7} &  \textbf{5193.7} &  \textbf{5193.7} &  \textbf{5193.7}\\
			\cline{2-12}
			&  \PB{} &  5184.0 &  5191.8 &  \textbf{5193.7} &  \textbf{5193.7} &  \textbf{5193.7} &  \textbf{5193.7} &  \textbf{5193.7} &  \textbf{5193.7} &  \textbf{5193.7} &  \textbf{5193.7}\\
			\cline{2-12}
			&  \MP{} &  \textbf{5184.3} &  \textbf{5192.9} &  \textbf{5193.7} &  \textbf{5193.7} &  \textbf{5193.7} &  \textbf{5193.7} &  \textbf{5193.7} &  \textbf{5193.7} &  \textbf{5193.7} &  \textbf{5193.7}\\
			\hline
			\multirow{5}{*}{\rot{CA-GrQc}} &  \CB{}~\cite{2023RahulCent} &  409.8 &  695.6 &  822.7 &  923.3 &  985.5 &  1033.2 &  1093.9 &  1165.4 &  1226.8 &  1274.0\\
			\cline{2-12}
			&  \LB{} &  353.7 &  629.4 &  770.0 &  860.3 &  929.7 &  992.2 &  1039.3 &  1101.7 &  1134.6 &  1183.1\\
			\cline{2-12}
			&  \MD{} &  289.2 &  563.1 &  770.2 &  936.8 &  1023.3 &  1091.1 &  1164.0 &  1232.2 &  1282.4 &  1323.9\\
			\cline{2-12}
			&  \PB{} &  479.4 &  806.7 &  922.7 &  1021.7 &  1119.0 &  1197.5 &  1260.4 &  1313.5 &  1372.6 &  1428.2\\
			\cline{2-12}
			&  \MP{} &  \textbf{641.8 }&  \textbf{863.0 }&  \textbf{1013.5} &  \textbf{1119.9} &  \textbf{1212.2} &  \textbf{1291.4} &  \textbf{1360.5} &  \textbf{1423.5} &  \textbf{1479.9} &  \textbf{1531.7}\\
			\hline
			\multirow{5}{*}{\rot{CA-HepPh}}  &  \CB{}~\cite{2023RahulCent} &  2361.5 &  2963.1 &  3281.7 &  3448.9 &  3697.0 &  3893.7 &  3960.5 &  4080.2 &  4158.2 &  4257.1\\
			\cline{2-12}
			&  \LB{} &  2753.6 &  3044.4 &  3245.9 &  3370.0 &  3491.2 &  3694.3 &  3838.4 &  3894.9 &  3949.9 &  4013.3\\
			\cline{2-12}
			&  \MD{} &  2692.1 &  2855.4 &  3229.2 &  3325.4 &  3438.7 &  3506.7 &  3661.1 &  3772.2 &  3922.8 &  4067.5\\
			\cline{2-12}
			&  \PB{} &  2547.1 &  2812.2 &  3325.6 &  3462.0 &  3624.0 &  3752.2 &  3860.2 &  3950.3 &  4072.9 &  4130.1\\
			\cline{2-12}
			&  \MP{} &  \textbf{3225.9} &  \textbf{3500.7} &  \textbf{3685.5} &  \textbf{3834.4} &  \textbf{3966.1} &  \textbf{4108.7} &  \textbf{4216.0} &  \textbf{4314.0} &  \textbf{4409.2} &  \textbf{4493.1}\\
			\hline
			\multirow{5}{*}{\rot{CA-HepTh}} &  \CB{}~\cite{2023RahulCent} &  723.1 &  1143.7 &  1351.7 &  1491.7 &  1601.6 &  1826.2 &  1944.1 &  2078.4 &  2158.9 &  2229.3\\
			\cline{2-12}
			&  \LB{} &  685.3 &  932.3 &  1214.0 &  1366.3 &  1635.3 &  1764.1 &  1854.6 &  1922.0 &  2000.1 &  2144.6\\
			\cline{2-12}
			&  \MD{} &  772.6 &  1086.3 &  1205.8 &  1510.9 &  1702.5 &  1819.4 &  1891.0 &  1953.9 &  2041.3 &  2136.8\\
			\cline{2-12}
			&  \PB{} &  776.7 &  1058.2 &  1323.7 &  1512.5 &  1714.8 &  1858.8 &  1981.3 &  2090.3 &  2184.7 &  2341.4\\
			\cline{2-12}
			&  \MP{} &  \textbf{943.0} &  \textbf{1349.8} &  \textbf{1638.3} &  \textbf{1843.2 }&  \textbf{2001.5} &  \textbf{2139.1} &  \textbf{2267.5} &  \textbf{2388.7} &  \textbf{2498.5} &  \textbf{2593.4}\\
			\hline
			\multirow{5}{*}{\rot{facebook}} &  \CB{}~\cite{2023RahulCent} &  \textbf{1987.5} &  \textbf{1987.5} &  \textbf{1987.5} &  \textbf{1987.5} &  \textbf{1987.5} &  \textbf{1987.5} &  \textbf{1987.5} &  \textbf{1987.5} &  \textbf{1987.5} &  \textbf{1987.5}\\
			\cline{2-12}
			&  \LB{} &  1938.8 &  1985.9 &  \textbf{1987.5} &  \textbf{1987.5} &  \textbf{1987.5} &  \textbf{1987.5} &  \textbf{1987.5} &  \textbf{1987.5} &  \textbf{1987.5} &\textbf{1987.5}\\
			\cline{2-12}
			&  \MD{} &  1939.3 &  \textbf{1987.5} &  \textbf{1987.5} &  \textbf{1987.5} &  \textbf{1987.5} &  \textbf{1987.5} &  \textbf{1987.5} &  \textbf{1987.5} &  \textbf{1987.5} &  \textbf{1987.5}\\
			\cline{2-12}
			&  \PB{} &  1987.5 &  \textbf{1987.5} &  \textbf{1987.5} &  \textbf{1987.5} &  \textbf{1987.5} &  \textbf{1987.5} &  \textbf{1987.5} &  \textbf{1987.5} &  \textbf{1987.5} &  \textbf{1987.5}\\
			\cline{2-12}
			&  \MP{} &  \textbf{1987.5} &  \textbf{1987.5} &  \textbf{1987.5} &  \textbf{1987.5} &  \textbf{1987.5} &  \textbf{1987.5} &  \textbf{1987.5} &  \textbf{1987.5} &  \textbf{1987.5} &  \textbf{1987.5}\\
			\hline
			\multirow{5}{*}{\rot{power}} &  \CB{}~\cite{2023RahulCent} &  97.0 &  169.1 &  250.4 &  320.1 &  393.3 &  463.4 &  537.0 &  604.3 &  656.0 &  697.1\\
			\cline{2-12}
			&  \LB{} &  102.5 &  187.6 &  244.0 &  315.3 &  372.1 &  423.3 &  500.6 &  571.6 &  623.6 &  663.2\\
			\cline{2-12}
			&  \MD{} &  119.5 &  199.3 &  292.9 &  359.7 &  407.6 &  453.9 &  521.4 &  583.0 &  648.9 &  712.5\\
			\cline{2-12}
			&  \PB{} &  126.3 &  246.1 &  342.1 &  443.5 &  532.7 &  612.2 &  693.1 &  759.3 &  830.7 &  896.4\\
			\cline{2-12}
			&  \MP{} &  \textbf{189.5} & \textbf{ 336.4} &  \textbf{459.5} &  \textbf{570.8} &  \textbf{671.4} &  \textbf{761.2} &  \textbf{847.3} &  \textbf{927.9} &  \textbf{1003.7} &  \textbf{1072.9}\\
			\hline 
		\end{tabular}
		\label{tab:threshold-related-interest}
	\end{table}

\begin{table}[!htbp]
		\caption{The sum of interest values associated with influenced vertices with mechanism $t_A(u)= \lceil deg(u)*(1-\eta(u)) \rceil$ (setting (ii)) under the LTM model. Note that the sum of interest values is in thousands.}
		\centering
		\scriptsize
		\begin{tabular}{  c c c c c c c c c c c c } 
			\hline
			\multirow{2}{*} {Dataset} & \multirow{2}{*}{Algo.} & \multicolumn{10}{c}{Seed Set Size $10-100$}\\
			\cline{3-12}
			&  & \textbf{10} & \textbf{20} & \textbf{30} & \textbf{40} & \textbf{50} & \textbf{60} & \textbf{70} & \textbf{80} & \textbf{90} & \textbf{100} \\
			\hline
			\multirow{2}{*}{{com-dblp.ungraph}}  &  \MD{} & 3.53 & 6.23 & 9.02 & 10.93 & 13.68 & 15.02 & 16.4 & 17.7 & 18.61& 20.89\\
			\cline{2-12}
			&  \PB{} & 3.72 & 6.01 & 9.27 & 10.96 & 12.14 & 14.52 & 15.47 & 17.11 & 18.07 & 19.06 \\
        \hline
   \multirow{2}{*}{{Email-EuAll}}  &  \MD{} & 22.36 & 33.37 & 39.91 & 45.83 & 51.59& 56.34 & 59.70 & 63.29 & 67.82 & 70.23\\
			\cline{2-12}
			&  \PB{} &  22.36 & 32.85 & 40.53 & 45.99 & 50.64& 55.07 & 59.27 & 63.21 & 66.55 & 69.69\\
			
			\hline 
		\end{tabular}
		\label{tab:threshold-related-interest_large}
	\end{table}

 \begin{table}[!htbp]
		\caption{The time taken by the heuristics are shown in the table with mechanism $t_A(u)= \lceil deg(u)*(1-\eta(u)) \rceil$ (setting (ii)) under the LTM model. Note that numbers are in seconds.}
		\centering
		\scriptsize
		\begin{tabular}{  c c c c c c c c c c c c } 
			\hline
			\multirow{2}{*} {Dataset} & \multirow{2}{*}{Algo.} & \multicolumn{10}{c}{Seed Set Size $10-100$}\\
			\cline{3-12}
			&  & \textbf{10} & \textbf{20} & \textbf{30} & \textbf{40} & \textbf{50} & \textbf{60} & \textbf{70} & \textbf{80} & \textbf{90} & \textbf{100} \\
			\hline
			\multirow{2}{*}{{com-dblp.ungraph}}  &  \MD{} & 4.48 & 9.77 & 15.76 & 22.73 & 31.14 & 39.61 & 49.18 & 59.38 & 71.81 & 88.07\\
			\cline{2-12}
			&  \PB{} & 18.25 & 52.47 & 100.24 & 161.75 & 237.76 & 329.2 & 435.31 & 560.04 & 700.2 & 899.97 \\
        \hline
   \multirow{2}{*}{{Email-EuAll}}  &  \MD{} & 3.41 & 7.3 & 11.88 & 17.14 & 25.72 & 34.48 & 46.2 & 57.24 & 67.32 & 81.78\\
			\cline{2-12}
			&  \PB{} & 10.28 & 30.86 & 59.42 & 89.78 & 138.87 & 201.87 & 260.99 & 346.89 & 437.59 & 528.47\\
			
			\hline 
		\end{tabular}
		\label{tab:threshold-related-interest_large_time}
	\end{table}

	\begin{figure}[!htbp]
		\centering
		\includegraphics[scale=0.30]{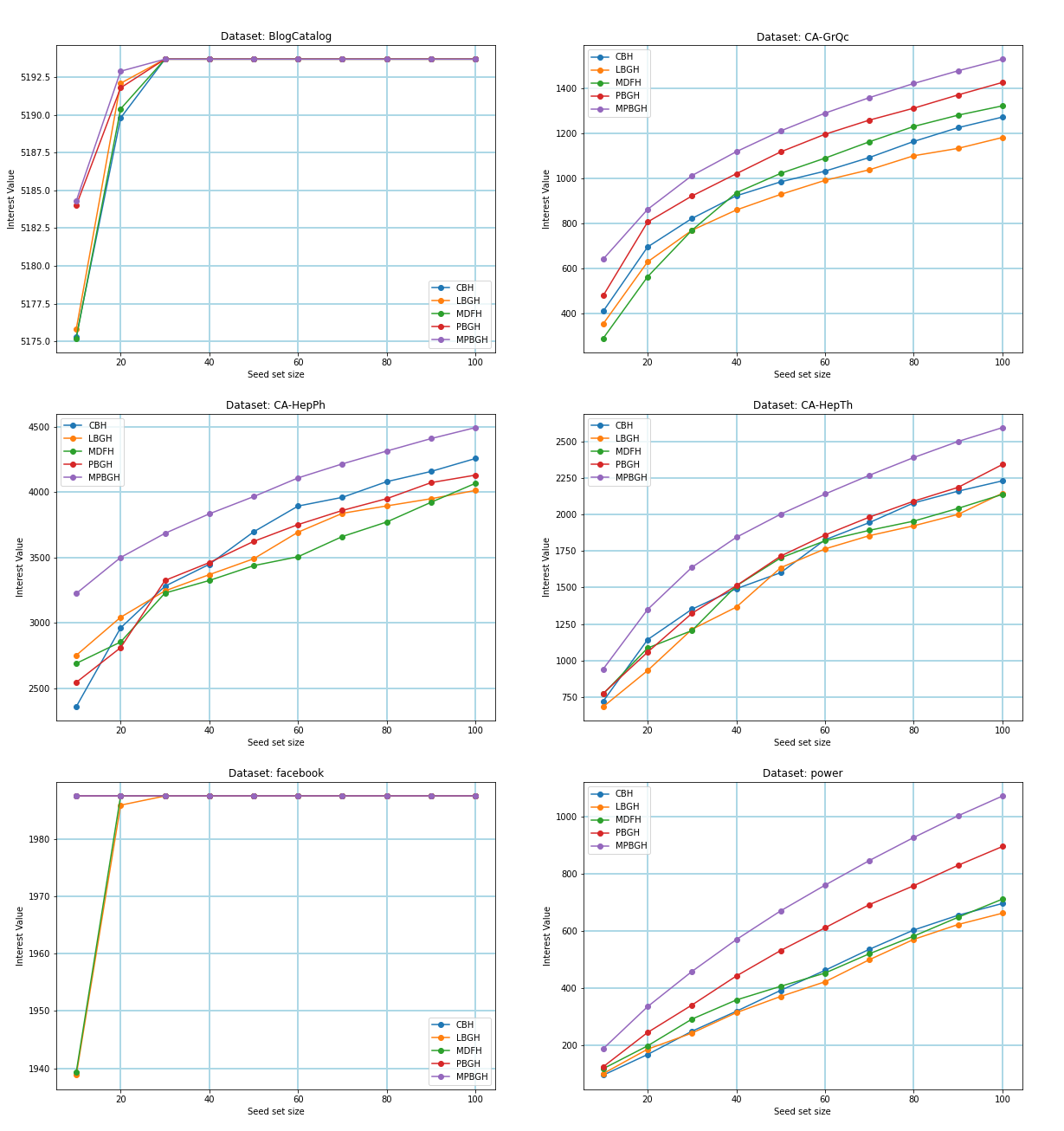}
		\caption{The sum of interest values associated with influenced vertices with mechanism $t_A(u)= \lceil deg(u)*(1-\eta(u)) \rceil$ (setting (ii)) under the LTM model. }
		\label{fig:threshold-related-interest}
	\end{figure}
	
	The \MPE{} (\MP{}) outperforms the remaining heuristics for all the real-world datasets. Note that \MP{} is computationally expensive. The other heuristic that competes well with \MP{} on all the data sets under both the threshold mechanisms is \PB{},  as can be seen in \figurename~\ref{fig:threshold-without-related-interest} and  \figurename~\ref{fig:threshold-related-interest}. Note that \PB{} is computationally much faster compared to \MP{}.
	For dense networks like BlogCatalog and Facebook, it can be seen that \PB{} comes much closer in performance to \MP{}. For the three datasets CA-HepPh, CA-GrQc, CA-HepTh, the \MP{} and  \PB{} are giving almost similar performance. The node $u$ selected by \MP{} has high diffusion strength, which may have a high-profit value. So  $u$ is also picked by \PB{}. As \PB{} is computationally much faster than \MP{}, this would turn out to be an advantage for \PB{}. But \MP{}'s performance is better than \PB{} for other datasets.
	
	For the threshold mechanism (ii), in the case of BlogCatalog and Facebook networks, the entire vertex set is influenced by $k > 20$ for all heuristics. Hence, they produce the same interest value as seen in \tablename~\ref{tab:threshold-related-interest}. Only in the case of Power data set, which has lower density and lower average degree, \MP{} outperforms all the other heuristics with a big gap in the interest value achieved as shown in  \tablename~\ref{tab:threshold-related-interest} and \PB{} does not compete well for datasets with the low average degree and density.

	
	\figurename~\ref{fig:threshold-related-interest} gives the results obtained for LTM under the threshold mechanism (ii), in which the threshold values are derived from the interest values of the vertices that influence or affect the propagation of the information in the networks. For all the networks,  \MP{} outperforms all the other heuristics. 
	\MP{} selects a seed vertex based on the diffusion strength of the vertex. Hence,  heuristic benefits from highly interested nodes located at more than one hop from the selected seed vertices.    All the heuristics \PB{}, \CB{}~\cite{2023RahulCent}, \MD{}, and \LB{} work well for the graphs with high average degree like BlogCatalog and facebook. 


\subsection{Interest Maximization under the Independent Cascade Model}
From an experimental point of view, we set the weights on the edge of the graph $G(V,E)$. The ICM is implemented with two different settings: $(i)$ $p(u,v)= 0.5$, and $(ii)$  $p(u,v)= 0.5*\eta(u)$.

As ICM is a probabilistic model, the heuristics (\LB{}, \MD{}, and \PB{}) are run $200$ times, and the average results are tabulated. As the time complexity of \MP{} is higher (~\ref{fig:time4}), so we have run this algorithm only $20$ times.

The results of the proposed algorithms under $(i)$ are shown in \tablename~\ref{tab:probability-independent-with-interest}, which are plotted in the \figurename~\ref{fig:probability-independent-with-interest}. Similarly, the results under setting $(ii)$ are depicted on \tablename~\ref{tab:probability-dependent-with-interest}, which are plotted in the \figurename~\ref{fig:probability-dependent-with-interest}.

For datasets like BlogCatalog and facebook,  which are dense and have higher average degrees, the \MP{} outperforms the other heuristics. Additionally, for sparse graphs like power, the gap between \MP{} and others increases with seed size until convergence. The \PB{} is in second place after \MP{}, but \PB{} is faster than \MP{} from the running time perspective.


Under mechanism (ii), the activation probability $p(u,v)=0.5*\eta(u)$ is in the range $(0,0.5]$ ; 
So, the information diffusion or the interest value achieved after diffusion is less than the first setting. In this case, also, \MP{} outperforms all the other heuristics. 

As shown in \figurename~\ref{fig:probability-independent-with-interest} and \figurename~\ref{fig:probability-dependent-with-interest}, for the sparse graph, \LB{} performs well compared to \MD{} and \PB{}. Overall, the \MP{} outperforms the other heuristics, but time complexity is higher than other heuristics because the \MP{} tests the strength of diffusion capability of each vertex by running the diffusion model and then selecting the seed vertex.

If we see the figures \figurename~\ref{fig:probability-independent-with-interest} and \figurename~\ref{fig:probability-dependent-with-interest} ( corresponding tables \tablename~\ref{tab:probability-independent-with-interest} and \tablename~\ref{tab:probability-dependent-with-interest}), the sum of interest value is less for the case $(ii)$($p(u,v)= 0.5*\eta(u)$) as compared to case $(i)$ ($0.5*\eta(u)$) for same seed size and dataset. The reason is value $0.5$ more than or equal to $0.5*eta(u)$ (where $\eta(u)\in[0,1]$). To activate a vertex from $u$, the activation probability should be less than or equal to $0.5$ or $0.5*\eta(u)$. It is very obvious that the first case $(i)$ ($0.5$) can activate more vertices that lead to a higher sum of interest value. Therefore, the sum of interest is higher in case of weight on edges is $0.5$.


\begin{table}[!htbp]
	\caption{ The sum of interest value associated with influenced vertices under the ICM where the weight on each edge is considered $0.5$.}
	\centering
	\scriptsize
	\begin{tabular}{ c c c c c c c c c c c c } 
		\hline
		\multirow{1}{*} {Dataset} & \multirow{1}{*}{Algo.} & \multicolumn{10}{c}{Seed Set Size $10-100$}\\
		\cmidrule{3-12}
		&  & \textbf{10} & \textbf{20} & \textbf{30} & \textbf{40} & \textbf{50} & \textbf{60} & \textbf{70} & \textbf{80} & \textbf{90} & \textbf{100} \\
		\hline
		\multirow{4}{*}{\rot{BlogCatalog}}  &  \LB{} &  5030.0 &  5032.6 &  5035.2 &  5037.8 &  5040.7 &  5043.3 &  5046.1 &  5049.1 &  5051.4 &  5054.0\\
        \cline{2-12}
         &  \MD{} &  5037.1 &  5041.9 &  5047.1 &  5052.1 &  5057.3 &  5062.6 &  5067.7 &  5073.0 &  5078.4 &  5083.7\\
        \cline{2-12}
         &  \PB{} &  5032.3 &  5041.8 &  5051.0 &  5060.0 &  5068.8 &  5077.2 &  5085.3 &  5093.0 &  5100.6 &  5107.8\\
        \cline{2-12}
		&  \MP{} & \textbf{ 5068.7} & \textbf{ 5077.9} & \textbf{ 5086.7} & \textbf{ 5095.2} & \textbf{ 5103.3} & \textbf{ 5111.1} & \textbf{ 5118.6} & \textbf{ 5125.8} & \textbf{ 5132.8} & \textbf{ 5139.3}\\
		\hline
		\multirow{5}{*}{\rot{CA-GrQc}}&  \LB{} &  1570.4 &  1595.9 &  1619.2 &  1640.5 &  1658.6 &  1676.8 &  1693.3 &  1712.3 &  1729.8 &  1745.8\\
\cline{2-12}
 &  \MD{} &  1555.0 &  1581.3 &  1603.6 &  1624.9 &  1642.7 &  1660.8 &  1678.4 &  1692.4 &  1703.9 &  1717.5\\
\cline{2-12}
 &  \PB{} &  1574.2 &  1606.2 &  1632.2 &  1656.2 &  1680.4 &  1700.5 &  1719.4 &  1737.7 &  1754.0 &  1769.4\\
\cline{2-12}
		&  \MP{} & \textbf{ 1647.0} & \textbf{ 1687.3} & \textbf{ 1723.0} & \textbf{ 1755.0} & \textbf{ 1784.5} & \textbf{ 1811.9} & \textbf{ 1837.7} & \textbf{ 1861.9} & \textbf{ 1884.6} & \textbf{ 1906.0}\\
		\hline
		\multirow{4}{*}{\rot{CA-HepPh}}  &  \LB{} &  4817.8 &  4824.7 &  4830.6 &  4837.1 &  4844.8 &  4851.7 &  4857.3 &  4861.6 &  4868.8 &  4873.0\\
\cline{2-12}
 &  \MD{} &  4825.6 &  4853.6 &  4874.4 &  4888.5 &  4905.2 &  4921.6 &  4936.4 &  4947.8 &  4958.0 &  4973.0\\
\cline{2-12}
 &  \PB{} &  4839.9 &  4872.3 &  4899.2 &  4922.2 &  4942.6 &  4960.7 &  4976.2 &  4991.6 &  5007.2 &  5021.7\\
\cline{2-12}
		&  \MP{} & \textbf{ 4917.5} & \textbf{ 4957.4} & \textbf{ 4991.7} & \textbf{ 5022.2} & \textbf{ 5050.6} & \textbf{ 5076.6} & \textbf{ 5100.0} & \textbf{ 5121.7} & \textbf{ 5141.9} & \textbf{ 5160.9}\\
		\hline
		\multirow{4}{*}{\rot{CA-HepTh}}  &  \LB{} &  3285.5 &  3297.7 &  3311.9 &  3326.3 &  3335.8 &  3343.8 &  3349.7 &  3357.3 &  3364.0 &  3372.7\\
\cline{2-12}
 &  \MD{} &  3317.2 &  3345.2 &  3366.6 &  3388.9 &  3412.3 &  3429.8 &  3443.8 &  3462.2 &  3478.9 &  3491.5\\
\cline{2-12}
 &  \PB{} &  3305.7 &  3341.3 &  3372.2 &  3399.1 &  3424.3 &  3446.3 &  3466.8 &  3485.8 &  3505.9 &  3523.1\\
\cline{2-12}
		&  \MP{} & \textbf{ 3412.8} & \textbf{ 3459.4} & \textbf{ 3500.5} & \textbf{ 3537.3} & \textbf{ 3570.5} & \textbf{ 3600.4} & \textbf{ 3628.7} & \textbf{ 3655.2} & \textbf{ 3680.0} & \textbf{ 3703.4}\\
		\hline
		\multirow{4}{*}{\rot{facebook} } &  \LB{} &  1947.1 &  1949.9 &  1952.5 &  1955.1 &  1958.9 &  1965.6 &  1973.0 &  1979.4 &  1983.3 &  1984.1\\
\cline{2-12}
 &  \MD{} &  1948.5 &  1954.0 &  1959.2 &  1964.8 &  1969.5 &  1975.2 &  1979.9 &  1983.9 &  1985.2 &  1985.5\\
\cline{2-12}
 &  \PB{} &  1953.0 &  1961.5 &  1968.8 &  1974.8 &  1979.7 &  1983.1 &  1985.5 &  1986.9 &  1987.4 &  1987.5\\
\cline{2-12}
		&  \MP{} & \textbf{ 1967.6} & \textbf{ 1974.8} & \textbf{ 1980.4} & \textbf{ 1984.3} & \textbf{ 1986.6} & \textbf{ 1987.5} & \textbf{ 1987.5} & \textbf{ 1987.5} & \textbf{ 1987.5} & \textbf{ 1987.5}\\
		\hline
		\multirow{4}{*}{\rot{power}} &  \LB{} &  419.6 &  554.6 &  680.1 &  736.9 &  780.0 &  812.5 &  859.2 &  888.1 &  912.4 &  939.6\\
\cline{2-12}
 &  \MD{} &  441.2 &  627.0 &  725.2 &  811.1 &  877.1 &  924.5 &  972.5 &  1008.8 &  1048.9 &  1083.8\\
\cline{2-12}
 &  \PB{} &  406.8 &  575.2 &  707.4 &  798.2 &  856.4 &  916.9 &  963.9 &  1004.9 &  1049.9 &  1090.0\\
\cline{2-12}
		&  \MP{} & \textbf{ 955.9} & \textbf{ 1176.3} & \textbf{ 1308.0} & \textbf{ 1406.9} & \textbf{ 1484.6} & \textbf{ 1547.4} & \textbf{ 1606.2} & \textbf{ 1651.7} & \textbf{ 1695.7} & \textbf{ 1733.7}\\
		\hline
		
	\end{tabular}
	\label{tab:probability-independent-with-interest}
\end{table}

\begin{figure}[!htbp]
	\centering
	\includegraphics[scale=0.30]{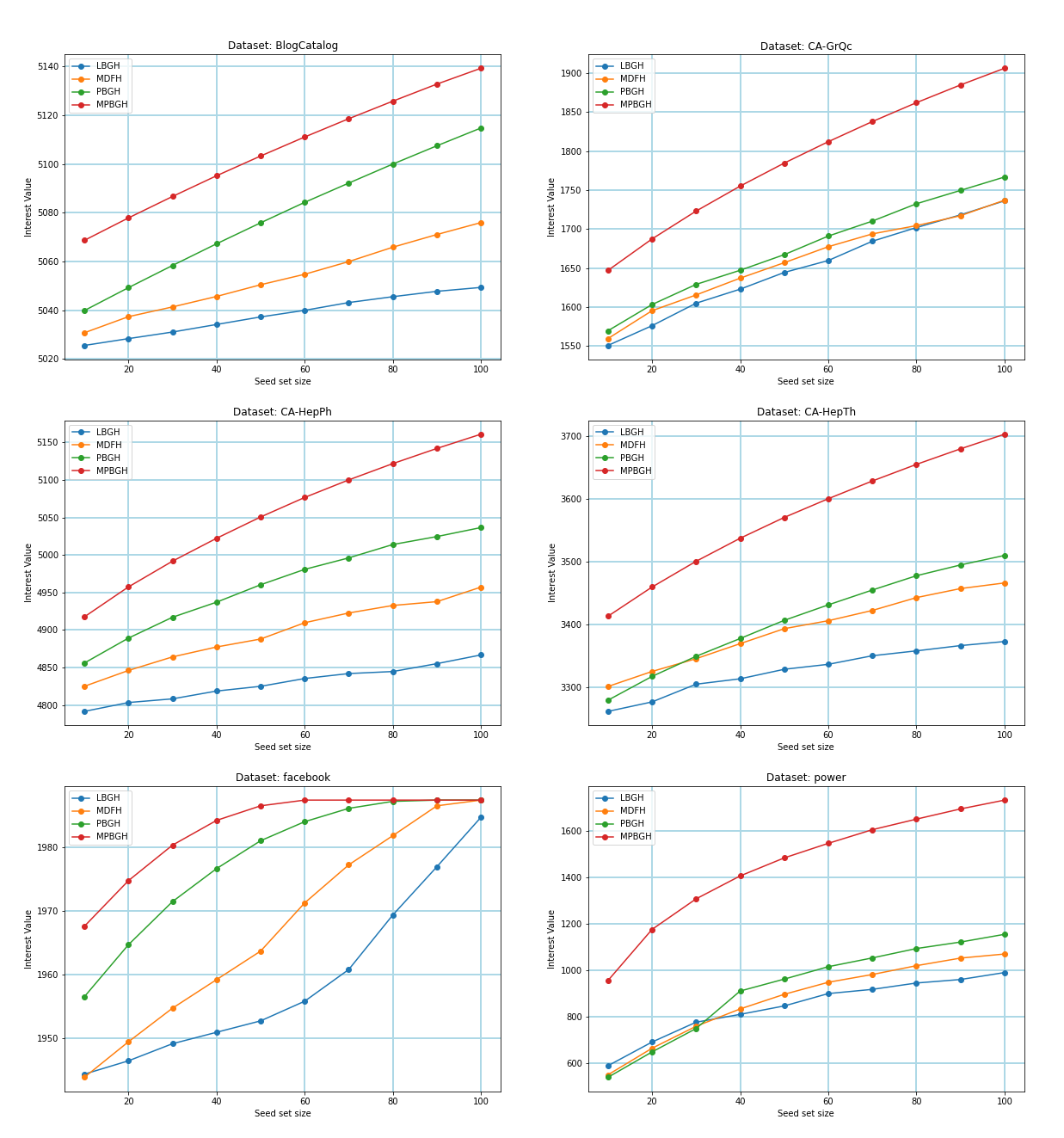}
	\caption{The sum of interest value associated with influenced vertices under the ICM setting (i) }
	\label{fig:probability-independent-with-interest}
\end{figure}


\begin{table}[!htbp]
	\caption{The sum of interest value associated with influenced vertices under the ICM where the weight on each edge is considered $p(u,v)= 0.5*\eta(u)$ ( setting (ii) ).}
	\centering
	\scriptsize
	\begin{tabular}{  c c c c c c c c c c c c } 
		\hline
		\multirow{2}{*} {Dataset} & \multirow{2}{*}{Algo.} & \multicolumn{10}{c}{Seed Set Size $10-100$}\\
		\cline{3-12}
		&  & \textbf{10} & \textbf{20} & \textbf{30} & \textbf{40} & \textbf{50} & \textbf{60} & \textbf{70} & \textbf{80} & \textbf{90} & \textbf{100} \\
		\hline
		\multirow{5}{*}{\rot{BlogCatalog}} &  \LB{} &  4727.8 &  4730.6 &  4733.0 &  4735.8 &  4738.5 &  4741.0 &  4743.6 &  4746.0 &  4748.3 &  4750.7\\
\cline{2-12}
 &  \MD{} &  4732.0 &  4736.6 &  4741.6 &  4746.9 &  4751.9 &  4757.4 &  4762.2 &  4766.8 &  4771.7 &  4777.0\\
\cline{2-12}
 &  \PB{} &  4740.1 &  4749.8 &  4759.5 &  4769.2 &  4778.8 &  4788.2 &  4797.6 &  4806.9 &  4816.1 &  4825.2\\
\cline{2-12}
		&  \MP{} & \textbf{ 4783.7} & \textbf{ 4793.7} & \textbf{ 4803.4} & \textbf{ 4813.1} & \textbf{ 4822.6} & \textbf{ 4831.9} & \textbf{ 4841.2} & \textbf{ 4850.5} & \textbf{ 4859.6} & \textbf{ 4868.6}\\
		\hline
		\multirow{4}{*}{\rot{CA-GrQc}} &  \LB{} &  868.1 &  904.0 &  930.4 &  960.4 &  986.9 &  1008.4 &  1027.3 &  1045.1 &  1064.4 &  1080.3\\
\cline{2-12}
 &  \MD{} &  836.3 &  869.0 &  891.1 &  909.3 &  925.2 &  944.3 &  960.0 &  976.7 &  994.1 &  1011.1\\
\cline{2-12}
 &  \PB{} &  885.8 &  925.2 &  951.3 &  975.0 &  993.3 &  1016.0 &  1033.3 &  1048.1 &  1063.8 &  1078.2\\
\cline{2-12}
		&  \MP{} & \textbf{ 1020.5} & \textbf{ 1085.2} & \textbf{ 1135.8} & \textbf{ 1182.2} & \textbf{ 1224.5} & \textbf{ 1262.6} & \textbf{ 1298.3} & \textbf{ 1331.4} & \textbf{ 1363.4} & \textbf{ 1392.2}\\
		\hline
		\multirow{4}{*}{\rot{CA-HepPh}}  &  \LB{} &  3552.4 &  3561.4 &  3571.1 &  3577.2 &  3580.7 &  3585.8 &  3594.9 &  3601.3 &  3605.3 &  3611.0\\
\cline{2-12}
 &  \MD{} &  3575.5 &  3597.3 &  3620.8 &  3635.8 &  3655.8 &  3666.6 &  3676.4 &  3688.8 &  3703.6 &  3715.7\\
\cline{2-12}
 &  \PB{} &  3605.1 &  3634.3 &  3662.0 &  3686.6 &  3710.1 &  3731.6 &  3748.0 &  3766.5 &  3783.9 &  3802.7\\
\cline{2-12}
		&  \MP{} & \textbf{ 3725.0} & \textbf{ 3785.2} & \textbf{ 3838.9} & \textbf{ 3887.8} & \textbf{ 3931.8} & \textbf{ 3970.9} & \textbf{ 4008.2} & \textbf{ 4044.0} & \textbf{ 4077.4} & \textbf{ 4110.0}\\
		\hline
		\multirow{4}{*}{\rot{CA-HepTh}}  &  \LB{} &  1883.2 &  1891.5 &  1901.7 &  1911.9 &  1920.2 &  1929.7 &  1938.9 &  1946.5 &  1952.8 &  1956.6\\
\cline{2-12}
 &  \MD{} &  1875.3 &  1916.3 &  1940.0 &  1962.2 &  1988.7 &  2013.7 &  2034.9 &  2049.9 &  2067.5 &  2086.2\\
\cline{2-12}
 &  \PB{} &  1895.6 &  1939.1 &  1975.2 &  2004.5 &  2029.7 &  2054.0 &  2076.7 &  2096.8 &  2118.5 &  2142.7\\
\cline{2-12}
		&  \MP{} & \textbf{ 2101.2} & \textbf{ 2204.7} & \textbf{ 2268.9} & \textbf{ 2334.2} & \textbf{ 2388.9} & \textbf{ 2438.4} & \textbf{ 2485.4} & \textbf{ 2530.2} & \textbf{ 2570.1} & \textbf{ 2609.0}\\
		\hline
		\multirow{4}{*}{\rot{facebook}} &  \LB{} &  1829.9 &  1837.3 &  1840.1 &  1842.5 &  1845.1 &  1847.9 &  1851.0 &  1853.9 &  1856.5 &  1859.8\\
\cline{2-12}
 &  \MD{} &  1823.1 &  1830.8 &  1838.2 &  1846.3 &  1854.7 &  1861.2 &  1868.1 &  1874.9 &  1882.3 &  1888.0\\
\cline{2-12}
 &  \PB{} &  1841.8 &  1852.3 &  1861.5 &  1868.5 &  1875.8 &  1884.3 &  1892.6 &  1900.9 &  1909.1 &  1916.6\\
\cline{2-12}
		&  \MP{} & \textbf{ 1874.6} & \textbf{ 1886.2} & \textbf{ 1895.7} & \textbf{ 1904.7} & \textbf{ 1913.2} & \textbf{ 1921.4} & \textbf{ 1928.9} & \textbf{ 1936.0} & \textbf{ 1942.7} & \textbf{ 1948.9}\\
		\hline
		\multirow{4}{*}{\rot{power}}&  \LB{} &  29.7 &  44.1 &  58.3 &  69.7 &  79.6 &  87.5 &  100.5 &  109.5 &  120.0 &  133.5\\
\cline{2-12}
 &  \MD{} &  42.2 &  85.0 &  110.7 &  131.4 &  153.3 &  168.9 &  181.8 &  207.0 &  233.6 &  250.4\\
\cline{2-12}
 &  \PB{} &  49.0 &  91.7 &  138.3 &  167.9 &  190.9 &  218.1 &  242.9 &  268.5 &  288.5 &  302.7\\
\cline{2-12}
		&  \MP{} & \textbf{ 201.4} & \textbf{ 354.7} & \textbf{ 469.9} & \textbf{ 579.8} & \textbf{ 667.9} & \textbf{ 743.0} & \textbf{ 815.9} & \textbf{ 879.9} & \textbf{ 935.2} & \textbf{ 988.9}\\
		\hline
	\end{tabular}
	\label{tab:probability-dependent-with-interest}
\end{table}

\begin{figure}[!htbp]
	\centering
	\includegraphics[scale=0.30]{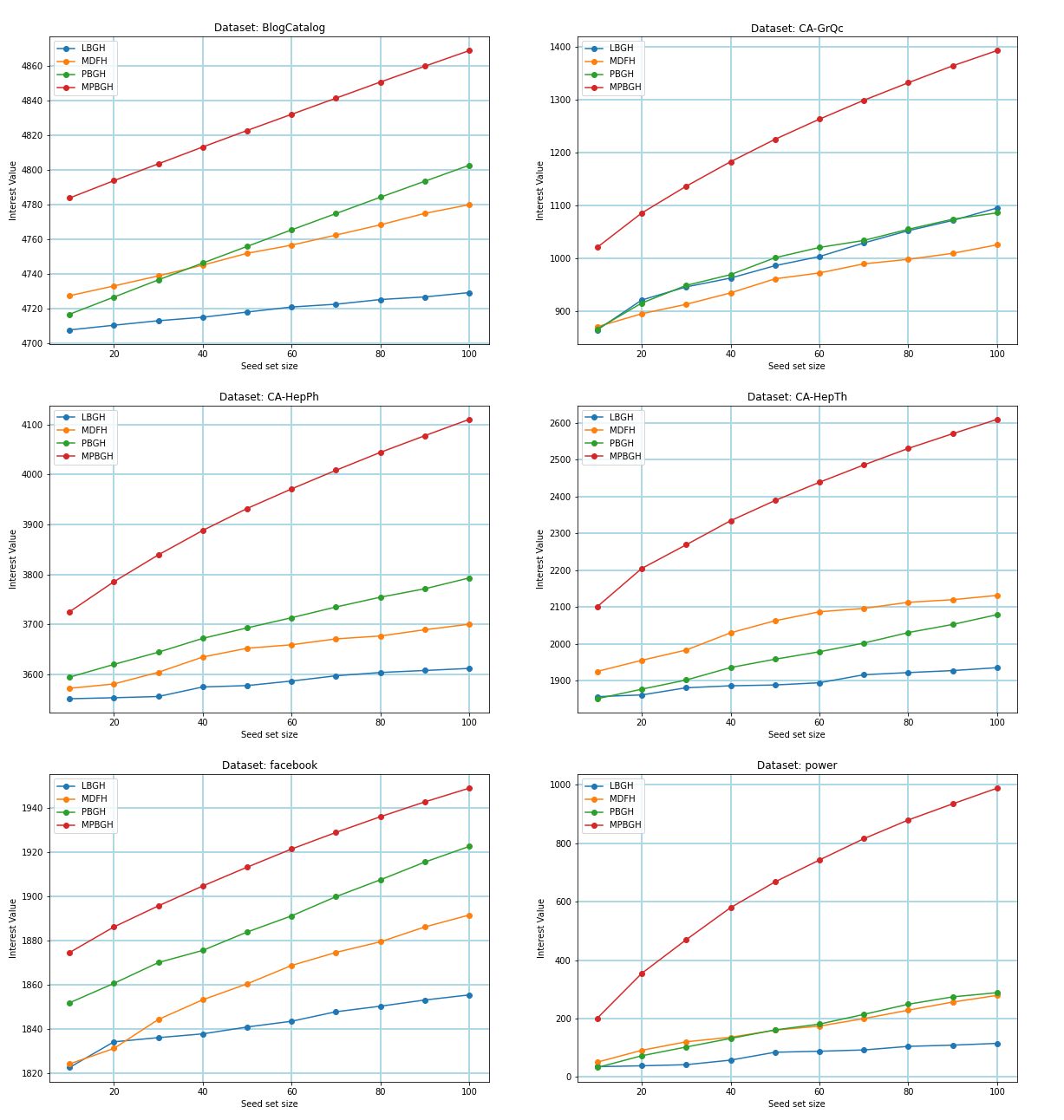}
	\caption{The sum of interest value associated with influenced vertices under the ICM, setting (ii).}
	\label{fig:probability-dependent-with-interest}
\end{figure}




\section{Performance Analysis}
   \figurename~\ref{fig:time2} and \figurename~\ref{fig:time4} are the graphs between computation time and seed set size. The heuristics are run multiple times and taken its average running time in \figurename~\ref{fig:time4}. The time taken by all heuristics for the ICM model is more than the LTM model. Because the chances of getting less than $0.5*\eta(u)$ probability is low.  Few vertices get activated, and non-activated vertices are checked again in the next iteration. On the other hand, the Algorithm activates more vertices in LTM, which leads to a few checks for the next seed vertex among non-activated vertices. So, the same algorithm run for LTM takes less time than for the ICM model. We have tested two heuristics (\MD{} and \PB{}) on large datasets which are tabulated in \tablename~\ref{tab:threshold-related-interest_large} and time taken by algorithms in \tablename~\ref{tab:threshold-related-interest_large_time}. As other algorithms took more than $20$ minutes, we did not put their results.  

   We compare the performance of \MPBF{} with the Gurbi`\cite{gurobi} of Linear Integer programming in \tablename~\ref{tab:perfromace-gurbi}. We ran the Gurbi~\cite{gurobi} for at most $6$ minutes and could not find the optimal solution for seed sets $5,10$ and$ 12$ on Jazz~\ref{tab:graphs} dataset. The \MP{} gives better than the Gurbi~\cite{gurobi} Solver's result within a second; similarly, \PB{} also outperforms in time perspective because the Gurbi~\cite{gurobi} solver stops with finding the optimal solution of linear equations. Overall, the \PB{} is highly scalable for large datasets with good results. But the quality of result perspective, \MP{} outperforms all other heuristic algorithms. For future improvement of results, the datasets and code are available on project~\citeauthor{intmax}~\cite{intmax}.

\begin{table}
    \centering
    \caption{For Jazz Dataset, Sum of interest of value associated with influenced vertices under linear threshold model with setting $t(u)=\lceil 0.5*degree(u) \rceil$.}
    \begin{tabular}{cccccc}
        Seed-Size & Gurbi~\cite{gurobi} & \LB{}  & \MD{}  & \PB{}  & \MD{} 
        \\ \hline
         5 & 86.64 & 79.7 & 75.8 & 89.1  &  89.1\\
         10 & 96.64 & 86.01 & 85.37 & 96.9 & 97 \\ 
         \hline
    \end{tabular}
    \label{tab:perfromace-gurbi}
\end{table}

\begin{figure}
    \centering
    \includegraphics[width=\linewidth]{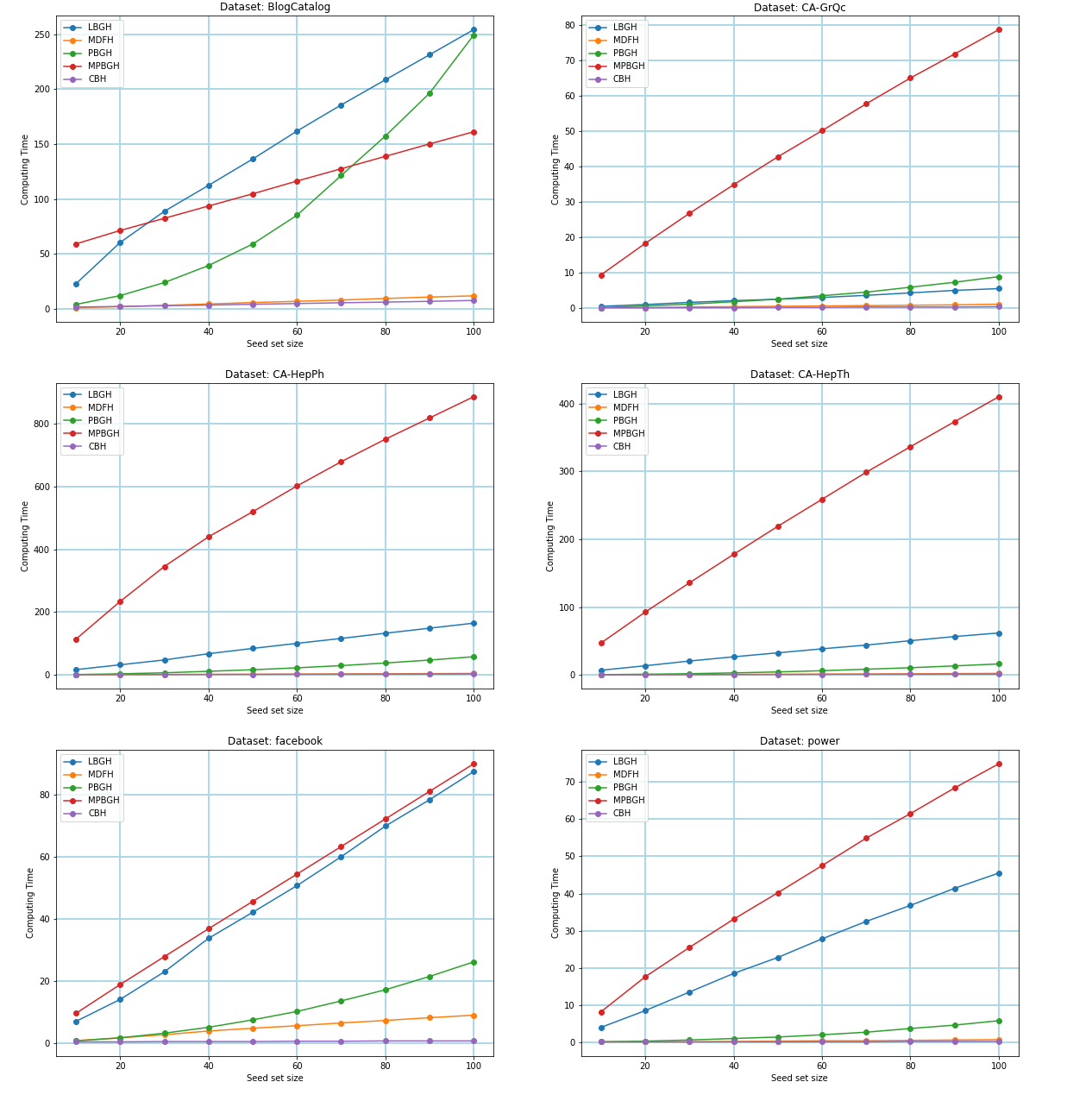}
    \caption{Computation time under LTM with setting $t_A(u)= \lceil deg(u)*(1-\eta(u)) \rceil$. }
    \label{fig:time2}
\end{figure}

\begin{figure}
    \centering
    \includegraphics[width=\linewidth]{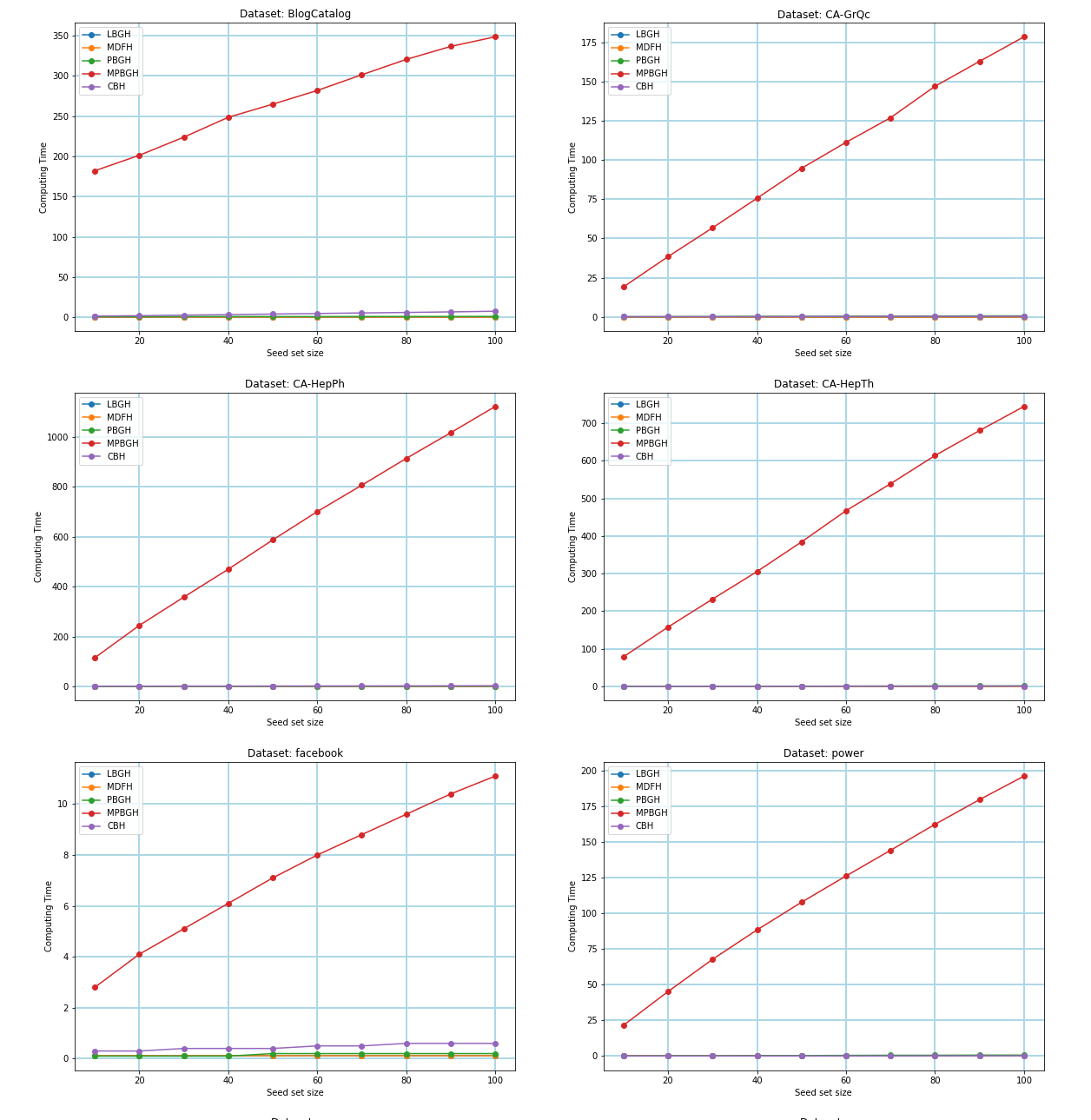}
    \caption{Computation time under ICM with setting $p(u,v)=0.5*\eta(u)$.}
    \label{fig:time4}
\end{figure}

\clearpage
\newpage
\section{Conclusion}
\label{sec:conclusion}
In this paper, we propose the problem of \IM{} on social networks.  The impact of people's interest in information propagation is studied. Small firms or companies can reach their highly interested customers through social media. The chances of selling products of the company get higher than randomly introducing products to all. We know a small company can not influence all people on social networks or does not need to influence all due to a limited product supply. Therefore, in this problem, for the given seed set size, the target is to maximize the sum of interest of influenced people. NP-Hardness and LP-formulation are proposed for \IM{}. From the experimental point of view, four heuristics (\LB{}, \MD{}, \PB{}, \MP{}) are presented in the paper and tested on $6$ data sets. We compare our heuristics to recent work ~\cite{2023RahulCent} under LTM and ICM. The \MPE{} outperforms all datasets. For dense graphs, the performance of \PB{} under LTM, which is computationally faster, is close to \MP{}. In future, one may look at problem is solvable for tree or not in polynomial time.   

\bibliography{sn-bibliography}


\begin{thebibliography}{23}
\ifx \bisbn   \undefined \def \bisbn  #1{ISBN #1}\fi
\ifx \binits  \undefined \def \binits#1{#1}\fi
\ifx \bauthor  \undefined \def \bauthor#1{#1}\fi
\ifx \batitle  \undefined \def \batitle#1{#1}\fi
\ifx \bjtitle  \undefined \def \bjtitle#1{#1}\fi
\ifx \bvolume  \undefined \def \bvolume#1{\textbf{#1}}\fi
\ifx \byear  \undefined \def \byear#1{#1}\fi
\ifx \bissue  \undefined \def \bissue#1{#1}\fi
\ifx \bfpage  \undefined \def \bfpage#1{#1}\fi
\ifx \blpage  \undefined \def \blpage #1{#1}\fi
\ifx \burl  \undefined \def \burl#1{\textsf{#1}}\fi
\ifx \doiurl  \undefined \def \doiurl#1{\url{https://doi.org/#1}}\fi
\ifx \betal  \undefined \def \betal{\textit{et al.}}\fi
\ifx \binstitute  \undefined \def \binstitute#1{#1}\fi
\ifx \binstitutionaled  \undefined \def \binstitutionaled#1{#1}\fi
\ifx \bctitle  \undefined \def \bctitle#1{#1}\fi
\ifx \beditor  \undefined \def \beditor#1{#1}\fi
\ifx \bpublisher  \undefined \def \bpublisher#1{#1}\fi
\ifx \bbtitle  \undefined \def \bbtitle#1{#1}\fi
\ifx \bedition  \undefined \def \bedition#1{#1}\fi
\ifx \bseriesno  \undefined \def \bseriesno#1{#1}\fi
\ifx \blocation  \undefined \def \blocation#1{#1}\fi
\ifx \bsertitle  \undefined \def \bsertitle#1{#1}\fi
\ifx \bsnm \undefined \def \bsnm#1{#1}\fi
\ifx \bsuffix \undefined \def \bsuffix#1{#1}\fi
\ifx \bparticle \undefined \def \bparticle#1{#1}\fi
\ifx \barticle \undefined \def \barticle#1{#1}\fi
\bibcommenthead
\ifx \bconfdate \undefined \def \bconfdate #1{#1}\fi
\ifx \botherref \undefined \def \botherref #1{#1}\fi
\ifx \url \undefined \def \url#1{\textsf{#1}}\fi
\ifx \bchapter \undefined \def \bchapter#1{#1}\fi
\ifx \bbook \undefined \def \bbook#1{#1}\fi
\ifx \bcomment \undefined \def \bcomment#1{#1}\fi
\ifx \oauthor \undefined \def \oauthor#1{#1}\fi
\ifx \citeauthoryear \undefined \def \citeauthoryear#1{#1}\fi
\ifx \endbibitem  \undefined \def \endbibitem {}\fi
\ifx \bconflocation  \undefined \def \bconflocation#1{#1}\fi
\ifx \arxivurl  \undefined \def \arxivurl#1{\textsf{#1}}\fi
\csname PreBibitemsHook\endcsname

\bibitem[\protect\citeauthoryear{Kempe et~al.}{2003}]{kempe2003maximizing}
\begin{bchapter}
\bauthor{\bsnm{Kempe}, \binits{D.}},
\bauthor{\bsnm{Kleinberg}, \binits{J.}},
\bauthor{\bsnm{Tardos}, \binits{{\'E}.}}:
\bctitle{Maximizing the spread of influence through a social network}.
In: \bbtitle{Proceedings of the Ninth ACM SIGKDD International Conference on
  Knowledge Discovery and Data Mining},
pp. \bfpage{137}--\blpage{146}
(\byear{2003})
\end{bchapter}
\endbibitem

\bibitem[\protect\citeauthoryear{India}{2020}]{rural2020}
\begin{botherref}
\oauthor{\bsnm{India}, \binits{T.O.}}:
For the first time, India has more rural net users than urban.
\url{https://timesofindia.indiatimes.com/business/india-business/for-the-first-time-india-has-more-rural-net-users-than-urban/articleshow/75566025.cms}
(2020)
\end{botherref}
\endbibitem

\bibitem[\protect\citeauthoryear{Cordasco
  et~al.}{2018}]{cordasco2018evangelism}
\begin{barticle}
\bauthor{\bsnm{Cordasco}, \binits{G.}},
\bauthor{\bsnm{Gargano}, \binits{L.}},
\bauthor{\bsnm{Rescigno}, \binits{A.A.}},
\bauthor{\bsnm{Vaccaro}, \binits{U.}}:
\batitle{Evangelism in social networks: Algorithms and complexity}.
\bjtitle{Networks}
\bvolume{71}(\bissue{4}),
\bfpage{346}--\blpage{357}
(\byear{2018})
\end{barticle}
\endbibitem

\bibitem[\protect\citeauthoryear{Cordasco et~al.}{2019}]{cordasco2019active}
\begin{barticle}
\bauthor{\bsnm{Cordasco}, \binits{G.}},
\bauthor{\bsnm{Gargano}, \binits{L.}},
\bauthor{\bsnm{Rescigno}, \binits{A.A.}}:
\batitle{Active influence spreading in social networks}.
\bjtitle{Theoretical Computer Science}
\bvolume{764},
\bfpage{15}--\blpage{29}
(\byear{2019})
\end{barticle}
\endbibitem

\bibitem[\protect\citeauthoryear{{Gurobi Optimization, LLC}}{2024}]{gurobi}
\begin{botherref}
\oauthor{\bsnm{{Gurobi Optimization, LLC}}}:
{Gurobi Optimizer Reference Manual}
(2024).
\url{https://www.gurobi.com}
\end{botherref}
\endbibitem

\bibitem[\protect\citeauthoryear{Kempe et~al.}{2005}]{kempe2005influential}
\begin{bchapter}
\bauthor{\bsnm{Kempe}, \binits{D.}},
\bauthor{\bsnm{Kleinberg}, \binits{J.M.}},
\bauthor{\bsnm{Tardos}, \binits{{\'E}.}}:
\bctitle{Influential nodes in a diffusion model for social networks.}
In: \bbtitle{ICALP},
vol. \bseriesno{5},
pp. \bfpage{1127}--\blpage{1138}
(\byear{2005}).
\bcomment{Springer}
\end{bchapter}
\endbibitem

\bibitem[\protect\citeauthoryear{Chen}{2009}]{chen2009approximability}
\begin{barticle}
\bauthor{\bsnm{Chen}, \binits{N.}}:
\batitle{On the approximability of influence in social networks}.
\bjtitle{SIAM Journal on Discrete Mathematics}
\bvolume{23}(\bissue{3}),
\bfpage{1400}--\blpage{1415}
(\byear{2009})
\end{barticle}
\endbibitem

\bibitem[\protect\citeauthoryear{Pereira et~al.}{2021}]{pereira2021effective}
\begin{barticle}
\bauthor{\bsnm{Pereira}, \binits{F.d.C.}},
\bauthor{\bsnm{Rezende}, \binits{P.J.}},
\bauthor{\bsnm{Souza}, \binits{C.C.}}:
\batitle{Effective heuristics for the perfect awareness problem}.
\bjtitle{Procedia Computer Science}
\bvolume{195},
\bfpage{489}--\blpage{498}
(\byear{2021})
\end{barticle}
\endbibitem

\bibitem[\protect\citeauthoryear{Gautam et~al.}{2023}]{2023RahulCent}
\begin{bchapter}
\bauthor{\bsnm{Gautam}, \binits{R.K.}},
\bauthor{\bsnm{Kare}, \binits{A.S.}},
\bauthor{\bsnm{Durga~Bhavani}, \binits{S.}}:
\bctitle{Centrality measures based heuristics for perfect awareness problem
  in social networks}.
In: \beditor{\bsnm{Morusupalli}, \binits{R.}},
\beditor{\bsnm{Dandibhotla}, \binits{T.S.}},
\beditor{\bsnm{Atluri}, \binits{V.V.}},
\beditor{\bsnm{Windridge}, \binits{D.}},
\beditor{\bsnm{Lingras}, \binits{P.}},
\beditor{\bsnm{Komati}, \binits{V.R.}} (eds.)
\bbtitle{Multi-disciplinary Trends in Artificial Intelligence},
pp. \bfpage{91}--\blpage{100}.
\bpublisher{Springer},
\blocation{Cham}
(\byear{2023})
\end{bchapter}
\endbibitem

\bibitem[\protect\citeauthoryear{Qiang et~al.}{2023}]{2023tieredInf}
\begin{botherref}
\oauthor{\bsnm{Qiang}, \binits{Z.}},
\oauthor{\bsnm{Pasiliao}, \binits{E.L.}},
\oauthor{\bsnm{Zheng}, \binits{Q.P.}}:
Target set selection in social networks with tiered influence and activation
  thresholds.
Journal of combinatorial optimization
\textbf{45:117}
(2023)
\doiurl{10.1007/s10878-023-01023-8}
\end{botherref}
\endbibitem

\bibitem[\protect\citeauthoryear{Ghosh et~al.}{2013}]{openion2013}
\begin{bbook}
\bauthor{\bsnm{Ghosh}, \binits{J.}},
\bauthor{\bsnm{Obradovic}, \binits{Z.}},
\bauthor{\bsnm{Dy}, \binits{J.}},
\bauthor{\bsnm{Zhou}, \binits{Z.-H.}},
\bauthor{\bsnm{Kamath}, \binits{C.}},
\bauthor{\bsnm{Parthasarathy}, \binits{S.}}:
\bbtitle{Proceedings of the 2013 SIAM International Conference on Data Mining
  (SDM)}.
\bpublisher{Society for Industrial and Applied Mathematics},
\blocation{Philadelphia, PA}
(\byear{2013}).
\doiurl{10.1137/1.9781611972832} .
\burl{https://epubs.siam.org/doi/abs/10.1137/1.9781611972832}
\end{bbook}
\endbibitem

\bibitem[\protect\citeauthoryear{Alla and Kare}{2023}]{alla2023opinion}
\begin{bchapter}
\bauthor{\bsnm{Alla}, \binits{L.S.}},
\bauthor{\bsnm{Kare}, \binits{A.S.}}:
\bctitle{Opinion maximization in signed social networks using centrality
  measures and clustering techniques}.
In: \bbtitle{Distributed Computing and Intelligent Technology: 19th
  International Conference, ICDCIT 2023, Bhubaneswar, India, January 18--22,
  2023, Proceedings},
pp. \bfpage{125}--\blpage{140}
(\byear{2023}).
\bcomment{Springer}
\end{bchapter}
\endbibitem

\bibitem[\protect\citeauthoryear{Tang et~al.}{2015}]{inf-martingale}
\begin{bchapter}
\bauthor{\bsnm{Tang}, \binits{Y.}},
\bauthor{\bsnm{Shi}, \binits{Y.}},
\bauthor{\bsnm{Xiao}, \binits{X.}}:
\bctitle{Influence maximization in near-linear time: A martingale approach}.
In: \bbtitle{Proceedings of the 2015 ACM SIGMOD International Conference on
  Management of Data}.
\bsertitle{SIGMOD '15},
pp. \bfpage{1539}--\blpage{1554}.
\bpublisher{Association for Computing Machinery},
\blocation{New York, NY, USA}
(\byear{2015}).
\doiurl{10.1145/2723372.2723734} .
\burl{https://doi.org/10.1145/2723372.2723734}
\end{bchapter}
\endbibitem

\bibitem[\protect\citeauthoryear{Gautam et~al.}{2022}]{gautam2022faster}
\begin{botherref}
\oauthor{\bsnm{Gautam}, \binits{R.K.}},
\oauthor{\bsnm{Kare}, \binits{A.S.}},
\oauthor{\bsnm{Bhavani}, \binits{S.D.}}:
Faster heuristics for graph burning.
Applied Intelligence,
1--11
(2022)
\end{botherref}
\endbibitem

\bibitem[\protect\citeauthoryear{Bhattacharya et~al.}{2022}]{k-center}
\begin{barticle}
\bauthor{\bsnm{Bhattacharya}, \binits{B.}},
\bauthor{\bsnm{Das}, \binits{S.}},
\bauthor{\bsnm{Dev}, \binits{S.R.}}:
\batitle{The weighted k-center problem in trees for fixed k}.
\bjtitle{Theoretical Computer Science}
\bvolume{906},
\bfpage{64}--\blpage{75}
(\byear{2022})
\doiurl{10.1016/j.tcs.2022.01.005}
\end{barticle}
\endbibitem

\bibitem[\protect\citeauthoryear{Liang et~al.}{2023}]{2023targetedCompt}
\begin{barticle}
\bauthor{\bsnm{Liang}, \binits{Z.}},
\bauthor{\bsnm{He}, \binits{Q.}},
\bauthor{\bsnm{Du}, \binits{H.}},
\bauthor{\bsnm{Xu}, \binits{W.}}:
\batitle{Targeted influence maximization in competitive social networks}.
\bjtitle{Information Sciences}
\bvolume{619},
\bfpage{390}--\blpage{405}
(\byear{2023})
\end{barticle}
\endbibitem

\bibitem[\protect\citeauthoryear{Yang et~al.}{2020}]{rumor2020containment}
\begin{barticle}
\bauthor{\bsnm{Yang}, \binits{L.}},
\bauthor{\bsnm{Li}, \binits{Z.}},
\bauthor{\bsnm{Giua}, \binits{A.}}:
\batitle{Containment of rumor spread in complex social networks}.
\bjtitle{Information Sciences}
\bvolume{506},
\bfpage{113}--\blpage{130}
(\byear{2020})
\end{barticle}
\endbibitem

\bibitem[\protect\citeauthoryear{Nazeri et~al.}{2023}]{2023burningNumber}
\begin{barticle}
\bauthor{\bsnm{Nazeri}, \binits{M.}},
\bauthor{\bsnm{Mollahosseini}, \binits{A.}},
\bauthor{\bsnm{Izadi}, \binits{I.}}:
\batitle{A centrality based genetic algorithm for the graph burning problem}.
\bjtitle{Applied Soft Computing}
\bvolume{144},
\bfpage{110493}
(\byear{2023})
\doiurl{10.1016/j.asoc.2023.110493}
\end{barticle}
\endbibitem

\bibitem[\protect\citeauthoryear{Newman}{2015}]{2015MarkNewMan}
\begin{botherref}
\oauthor{\bsnm{Newman}, \binits{M.}}:
Network data.
\url{http://www-personal.umich.edu/~mejn/netdata/}
(2015)
\end{botherref}
\endbibitem

\bibitem[\protect\citeauthoryear{Reza and Huan}{2009}]{reza2009social}
\begin{botherref}
\oauthor{\bsnm{Reza}, \binits{Z.}},
\oauthor{\bsnm{Huan}, \binits{L.}}:
Social computing data repository.
\url{http://datasets.syr.edu/pages/datasets.html}
(2009)
\end{botherref}
\endbibitem

\bibitem[\protect\citeauthoryear{Rossi and Ahmed}{2015}]{nr}
\begin{bchapter}
\bauthor{\bsnm{Rossi}, \binits{R.A.}},
\bauthor{\bsnm{Ahmed}, \binits{N.K.}}:
\bctitle{The network data repository with interactive graph analytics and
  visualization}.
In: \bbtitle{AAAI}
(\byear{2015}).
\bcomment{\url{https://networkrepository.com}}
\end{bchapter}
\endbibitem

\bibitem[\protect\citeauthoryear{Leskovec and Krevl}{2014}]{snapnets}
\begin{botherref}
\oauthor{\bsnm{Leskovec}, \binits{J.}},
\oauthor{\bsnm{Krevl}, \binits{A.}}:
{SNAP Datasets}: {Stanford} Large Network Dataset Collection.
\url{http://snap.stanford.edu/data}
(2014)
\end{botherref}
\endbibitem

\bibitem[\protect\citeauthoryear{Gautam}{2024}]{intmax}
\begin{botherref}
\oauthor{\bsnm{Gautam}, \binits{R.K.}}:
{Interest Maximization in Social Networks}
(2024).
\url{https://github.com/RKG-kumar/interest-maximization}
\end{botherref}
\endbibitem

\end{thebibliography}

\end{document}